\documentclass[11pt,letterpaper]{article}
\usepackage{amsmath,amsthm,amssymb}
\usepackage{fullpage}
\usepackage[parfill]{parskip}
\usepackage[dvipsnames]{xcolor}
\usepackage{libertine}
\usepackage{tikz}
\usetikzlibrary{calc}
\usetikzlibrary{arrows.meta}
\usepackage{thm-restate}

\usepackage{hyperref}
\usepackage[nameinlink, capitalise]{cleveref}

\usepackage{tcolorbox}
\tcbset{
  myframe/.style={
    colback=gray!5,
    colframe=black,
    arc=4pt,
    outer arc=4pt,
    boxsep=5pt,
    left=3pt,
    right=5pt,
    top=5pt,
    bottom=5pt,
    enhanced,
    sharp corners
  }
}

\newtheorem{theorem}{Theorem}[section]
\newtheorem{lemma}[theorem]{Lemma}

\newtheorem{corollary}[theorem]{Corollary}

\theoremstyle{definition}
\newtheorem{definition}[theorem]{Definition}

\AddToHook{env/lemma/begin}{\crefalias{theorem}{lemma}}
\AddToHook{env/conjecture/begin}{\crefalias{theorem}{conjecture}}
\AddToHook{env/corollary/begin}{\crefalias{theorem}{corollary}}
\AddToHook{env/proposition/begin}{\crefalias{theorem}{proposition}}
\AddToHook{env/definition/begin}{\crefalias{theorem}{definition}}
\AddToHook{env/example/begin}{\crefalias{theorem}{example}}
\AddToHook{env/remark/begin}{\crefalias{theorem}{remark}}
\AddToHook{env/question/begin}{\crefalias{theorem}{question}}
\AddToHook{env/condition/begin}{\crefalias{theorem}{condition}}

\crefname{theorem}{Theorem}{Theorems}
\crefname{lemma}{Lemma}{Lemmas}
\crefname{conjecture}{Conjecture}{Conjectures}
\crefname{corollary}{Corollary}{Corollaries}
\crefname{proposition}{Proposition}{Propositions}
\crefname{definition}{Definition}{Definitions}
\crefname{example}{Example}{Examples}
\crefname{remark}{Remark}{Remarks}
\crefname{question}{Question}{Questions}
\crefname{condition}{Condition}{Conditions}

\crefname{Program}{Program}{Programs}
\creflabelformat{Program}{(#2\textup{#1})#3}

\newcommand{\E}{\mathbb{E}}

\newcommand{\alg}{\mathsf{ALG}}
\newcommand{\opt}{\mathsf{OPT}}
\newcommand{\st}{\text{start}}
\newcommand{\ed}{\text{end}}
\newcommand{\len}{\text{len}}

\newenvironment{proofof}[1]{{\vspace*{5pt} \noindent\bf Proof of #1:  }}{\hfill\rule{2mm}{2mm}}

\definecolor{linkc}{rgb}{0.6, 0.2, 0.3}
\definecolor{citec}{rgb}{0.3, 0.2, 0.6}
\definecolor{urlc}{rgb}{0.2, 0.6, 0.3}
\hypersetup{
    colorlinks=true,
    linkcolor=linkc,
    citecolor=citec,
    urlcolor=urlc
}

\title{Online Matching in Convex Bipartite Graphs}
\author{Yilong Feng \and Zhihao Gavin Tang \and Kangning Wang \and Xiaowei Wu}
\date{\today}

\begin{document}

\maketitle
\thispagestyle{empty}
\setcounter{page}{0}

\begin{abstract}
Online resource-allocation systems, like outpatient scheduling and spectrum allocation, often assign sequentially arriving requests to an ordered pool of scarce resources, where each request accepts a contiguous interval of feasible options. We study the resulting online matching problem on convex bipartite graphs under irrevocable decisions and adversarial arrivals. We first show that convexity alone does not improve the classic worst-case guarantee of $1-1/e$, achieved by \textsc{Ranking}. We then consider the uniform-length model, in which every online request has exactly $d$ consecutive offline neighbors. We propose \textsc{Flip}, which uses one random bit to commit ex-ante to either earliest-feasible assignment or latest-feasible assignment. Although either natural deterministic policy can waste capacity and be asymptotically only $1/2$-competitive, we show that their randomized mixture is $2/3$-competitive. This guarantee is tight for \textsc{Flip} and remains valid against a semi-adaptive adversary that observes the selected policy before choosing the arrival order. We also prove that no randomized online algorithm can achieve a competitive ratio strictly larger than $3/4$ in the uniform-length model.
\end{abstract}

%\renewcommand{\contentsname}{Contents}
% \tableofcontents

\newpage

\section{Introduction}

Online bipartite matching is a fundamental model for allocating resources to sequentially arriving requests.
It has been studied extensively since the seminal work of Karp, Vazirani, and Vazirani~\cite{conf/stoc/KarpVV90}.
A set of offline resources is known in advance, whereas online requests arrive one at a time and reveal their compatible resources only upon arrival.
Each request must be matched immediately and irrevocably to an available compatible resource, or be left unmatched.
The objective is to maximize the size of the final matching.
The performance of an online algorithm is typically measured by its competitive ratio, defined as the worst-case ratio between its expected matching size and the size of a maximum matching.
For general bipartite graphs, any greedy algorithm that matches a request whenever possible is $1/2$-competitive, and no deterministic algorithm can guarantee a better ratio.
The celebrated \textsc{Ranking} algorithm achieves the optimal randomized ratio of $1-1/e$ against the standard oblivious adversary~\cite{conf/stoc/KarpVV90}.

These worst-case guarantees make no structural assumptions about the compatibility graph.
In many operations-research applications, however, the offline resources possess a natural order and the compatibility constraints have additional structure.
The offline side represents a fixed stock of resources, and the system operator seeks to use this capacity efficiently by serving as many requests as possible.
The online side represents demand revealed sequentially: each request seeks an acceptable resource, but its feasible choices are restricted by operational requirements or service preferences.
Because decisions are irrevocable and future requests are unknown, each assignment carries an opportunity cost; choosing which compatible resource to allocate is therefore central to the overall service level and system throughput.

Consider, for example, a platform that allocates a fixed set of time slots over a planning horizon.
Each requester is willing to accept any slot within its acceptable time window, while the platform seeks to accommodate as many requests as possible.
This captures outpatient scheduling, where hospitals allocate remaining appointment capacity to patients with time preferences~\cite{journals/ior/FeldmanLTZ14}.
The same pattern appears in airport slot allocation, where an airline requests a preferred takeoff or landing time but may accept another time within a displacement window~\cite{journals/tra/ZografosAM18,journals/transci/FairbrotherZG20}.
A similar setting arises in spectrum allocation.
A licensed frequency band is divided into ordered channels, and assignments must account for channel availability and interference between nearby frequencies~\cite{journals/anor/AardalHKMS07}.
Spectrum requests may also arrive online and require immediate, non-preemptive admission decisions~\cite{journals/tc/XuWL10}.
In a simple stylized model, each device has a preferred channel and can use that channel or a fixed number of nearby channels.
If the same displacement limit applies to every device, all acceptable frequency ranges have the same length.
In all these applications, the resources are ordered, and the resources acceptable to each request are consecutive in that order.

Compatibility graphs with this consecutive-neighborhood property belong to the class of convex bipartite graphs~\cite{journals/Glover1967}.
A bipartite graph is convex on one side if the vertices on that side admit a linear order in which the neighborhood of every vertex on the opposite side is consecutive.
We assume convexity on the offline side, and the offline order is known to the online algorithm in advance.
We refer to the resulting online neighborhoods as \emph{interval neighborhoods}.

This gives rise to the following natural question:
\begin{center}
\emph{Can convexity on the offline side be exploited to beat $1-1/e$ for online bipartite matching?}
\end{center}

\subsection{Our Results}

Perhaps surprisingly, we show that convexity alone does not improve the worst-case service guarantee.
When request intervals may have arbitrary lengths, the optimal competitive ratio remains $1-1/e$.

\begin{theorem}
\label{theorem:arbitrary-length}
   No algorithm has a competitive ratio strictly larger than $1-1/e$ for online matching in convex bipartite graphs.
\end{theorem}

Recall that \textsc{Ranking} is $(1-1/e)$-competitive on general bipartite graphs, so this upper bound is tight.
At a high level, our construction reproduces the upper-triangular structure of the classic $1-1/e$ hard instance by recursively replacing each offline resource with a sufficiently large block of consecutive resources.
Every request then has an interval of feasible resources, but these intervals have widely varying lengths.
From an operations perspective, this result shows that an ordered capacity structure does not by itself protect the system from worst-case congestion: sufficiently heterogeneous flexibility across requests can recreate the same bottlenecks as in an unstructured compatibility graph.

Many applications, however, standardize the amount of flexibility offered to each request.
A scheduling platform may require every acceptable time window to contain the same number of slots; similarly, an airport or spectrum-allocation system may impose a common displacement tolerance.
Besides simplifying the service design, such a restriction provides a basic form of procedural fairness across customers.
Motivated by these settings, we study the \emph{uniform-length} model, in which every request can be assigned to exactly $d$ consecutive resources, for some $d\geq 2$.
We show that this restriction fundamentally changes the attainable service level: the $1-1/e$ barrier can be surpassed by a transparent randomized policy using only one random bit.

To explain the policy, recall a useful offline property of convex bipartite graphs.
For each online vertex $v$, let $s(v)$ and $l(v)$ denote the smallest and largest indices in its interval neighborhood, respectively.
When all requests are known in advance, a maximum matching can be obtained by processing requests in non-decreasing order of $l(v)$ and assigning each request to its smallest-indexed available resource.
Symmetrically, one may process requests in non-increasing order of $s(v)$ and assign each request to its largest-indexed available resource~\cite{journals/Glover1967}.
In the online system, the arrival sequence is exogenous, but the operator can retain the corresponding resource-selection decisions.
This yields two natural dispatching policies:
\begin{itemize}
    \item The \emph{up-matching rule} is an \emph{earliest-feasible assignment} policy: assign each arriving request to its smallest-indexed currently available feasible resource.
    \item The \emph{down-matching rule} is a \emph{latest-feasible assignment} policy: assign each arriving request to its largest-indexed currently available feasible resource.
\end{itemize}

Both policies are work-conserving in the sense that they serve a request whenever any compatible capacity remains, and neither policy requires demand forecasts or knowledge of future arrivals.
Nevertheless, each policy has a systematic directional bias and can waste capacity badly: it may consume resources that are critical for later requests while leaving residual capacity that subsequent demand cannot use.
Each rule always produces a maximal matching and is therefore $1/2$-competitive.
The examples in Figure~\ref{figure:up-down-hard-instances} show that this poor performance guarantee is essentially tight.
In each instance, the indicated policy serves only $d$ requests, whereas a clairvoyant scheduler serves $2d-1$.
The resulting ratio is $\frac{d}{2d-1}$, which converges to $1/2$ as $d\to\infty$.
Thus, neither earliest-feasible nor latest-feasible assignment provides a satisfactory worst-case service guarantee when used alone.

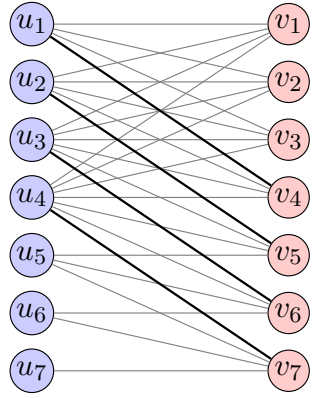
\begin{figure}[htb]
\centering
\resizebox{0.75\textwidth}{!}{%
\begin{tikzpicture}[
    x=1cm,
    y=0.72cm,
    vertex/.style={
        circle,
        draw,
        fill=white,
        minimum size=0.48cm,
        inner sep=1pt
    },
    u_vertex/.style={vertex, fill=blue!20},
    v_vertex/.style={vertex, fill=red!20},
    edge/.style={draw, gray, thin},
    match_edge/.style={draw, black, thick}
]

% -----------------------------------------------------------------
% Left panel: down-matching rule.
% -----------------------------------------------------------------
\node at (1.6,7.2) {Down-matching (latest feasible)};

\foreach \i in {1,...,7} {
    \node[u_vertex] (lu\i) at (0,{7-\i}) {$u_\i$};
    \node[v_vertex] (lv\i) at (3.2,{7-\i}) {$v_\i$};
}

\foreach \j in {1,...,4} {
    \draw[edge] (lv1) -- (lu\j);
}
\foreach \j in {2,...,5} {
    \draw[edge] (lv2) -- (lu\j);
}
\foreach \j in {3,...,6} {
    \draw[edge] (lv3) -- (lu\j);
}
\foreach \i in {4,...,7} {
    \foreach \j in {4,...,7} {
        \draw[edge] (lv\i) -- (lu\j);
    }
}

% Edges selected by the down-matching rule.
\draw[match_edge] (lv1) -- (lu4);
\draw[match_edge] (lv2) -- (lu5);
\draw[match_edge] (lv3) -- (lu6);
\draw[match_edge] (lv4) -- (lu7);

% Divider.
\draw[densely dashed, gray] (4.8,-0.4) -- (4.8,8);

% -----------------------------------------------------------------
% Right panel: up-matching rule.
% -----------------------------------------------------------------
\node at (8,7.2) {Up-matching (earliest feasible)};

\foreach \i in {1,...,7} {
    \node[u_vertex] (ru\i) at (6.4,{7-\i}) {$u_\i$};
    \node[v_vertex] (rv\i) at (9.6,{7-\i}) {$v_\i$};
}

\foreach \i in {1,...,4} {
    \foreach \j in {1,...,4} {
        \draw[edge] (rv\i) -- (ru\j);
    }
}
\foreach \j in {2,...,5} {
    \draw[edge] (rv5) -- (ru\j);
}
\foreach \j in {3,...,6} {
    \draw[edge] (rv6) -- (ru\j);
}
\foreach \j in {4,...,7} {
    \draw[edge] (rv7) -- (ru\j);
}

% Edges selected by the up-matching rule.
\draw[match_edge] (rv4) -- (ru1);
\draw[match_edge] (rv5) -- (ru2);
\draw[match_edge] (rv6) -- (ru3);
\draw[match_edge] (rv7) -- (ru4);

\end{tikzpicture}
}
\caption{Hard instances for the latest-feasible (down-matching) and earliest-feasible (up-matching) policies when $d=4$.
Blue and red nodes represent the offline and online vertices, respectively.
The arrival order for the down-matching rule is $v_1,v_2,\ldots,v_7$, while that for the up-matching rule is $v_7,v_6,\ldots,v_1$.
In both instances, the edges $(u_i,v_i)$, $i=1,\ldots,7$, form a perfect matching, whereas the indicated rule matches only $4$ online vertices.
}
\label{figure:up-down-hard-instances}
\end{figure}

The two instances also reveal the complementarity of the policies.
On the left instance, latest-feasible assignment wastes substantial capacity, whereas earliest-feasible assignment attains the offline optimum; on the right instance, the roles are reversed.
We therefore propose \textsc{Flip}, which makes a single ex-ante random choice between the two policies.
Before any request arrives, \textsc{Flip} samples a fair coin, commits to either the earliest-feasible or latest-feasible assignment, and follows that policy throughout the planning horizon.
Thus, the algorithm randomizes over two stable dispatching policies rather than making a new randomized decision for every request.
It uses only one random bit and does not require $d$ as an input.

This simple policy mixture provides a robust service-level guarantee.
For every uniform-length instance, its expected number of served requests is at least $\frac{2d-1}{3d-2}$ times the number served by a clairvoyant optimal scheduler.
Moreover, our analysis holds against a \emph{semi-adaptive adversary}: the compatibility graph is fixed first, the adversary then observes which of the two policies \textsc{Flip} has selected, and only afterward chooses the arrival order.
The adversary may therefore use a different policy-aware worst-case sequence for each outcome of the coin flip.
Consequently, the guarantee does not rely on concealing the selected dispatching rule; the operator may announce and consistently apply the policy before operations begin.
This public-commitment interpretation is stronger than the standard oblivious-adversary model and is particularly relevant in settings where assignment rules must be transparent to customers or regulators.
The policy is also operationally lightweight:
A standard balanced ordered-set data structure implements either policy in $O(\log m)$ time per arrival using $O(m)$ space, where $m$ is the number of offline resources.
Hence \textsc{Flip} is not only analytically simple but also easy to deploy at scale.

\begin{theorem}
    Algorithm \textsc{Flip} is $\frac{2d-1}{3d-2}$-competitive for online matching in uniform-length convex bipartite graphs, even against the semi-adaptive adversary described above, and the ratio is tight for \textsc{Flip}.
\end{theorem}

To benchmark the value of more sophisticated online controls, we complement this guarantee with an impossibility result that applies to every randomized assignment policy in the uniform-length model.

\begin{theorem}
    No algorithm has a competitive ratio strictly larger than $3/4$ for online matching in uniform-length convex bipartite graphs.
\end{theorem}

Consequently, the optimal worst-case service ratio in the uniform-length setting lies in $[2/3,3/4]$.
The lower endpoint is achieved asymptotically by the one-bit \textsc{Flip} policy, while the upper endpoint limits what any online policy can guarantee.

\subsection{Our Techniques}

We next outline the proof of the competitive guarantee for \textsc{Flip}.

\paragraph{From Oblivious Adversary to Semi-adaptive Adversary.}

Let $M^\star$ be the maximum matching derived by the offline greedy procedure that processes online vertices in non-decreasing order of $l(v)$~\cite{journals/Glover1967}.
We first show that it suffices to consider the instances where $M^\star$ covers all $n$ online vertices.
Thus $|M^\star|=n$.
We index the online vertices as $v_1,v_2,\ldots,v_n$ according to the order in which the offline greedy procedure processes them.
We analyze the performance of \textsc{Flip} against a semi-adaptive adversary\footnote{As Figure~\ref{figure:up-down-hard-instances} shows, \textsc{Flip} is only $1/2$-competitive against a fully-adaptive adversary.} that first fixes the underlying graph, then sees the rule chosen by the algorithm, and finally decides the online arrival order.
In particular, it may choose different arrival orders for the two outcomes of the random bit.
This adversary is stronger than the standard oblivious adversary. Thus, any guarantee against it also holds against the oblivious adversary.
We show that among all arrival orders, the up-matching rule attains its minimum matching size on the order $v_n,v_{n-1},\ldots,v_1$, whereas the down-matching rule attains its minimum on the order $v_1,v_2,\ldots,v_n$.
It therefore suffices to analyze these two canonical orders.

\paragraph{Encoding Unmatched Vertices to Bad Indices of an Offset Function.}
The maximum matching $M^\star$ allows us to encode the instance by a one-dimensional discrete function $f$.
Let $s_i$ be the smallest index of the $i$-th interval and let $\mu_i$ be the index of the offline vertex matched to it by $M^\star$.
Let $f(i):=\mu_i-s_i+1$, which denotes the number of neighbors of $v_i$ with index at most $\mu_i$.
We call $f$ the offset function, as $f(i)$ records the offset between the highest-priority neighbor of $v_i$ under the up-matching rule and the perfect neighbor of $v_i$ in $M^\star$.
By definition, we have $f(i)\in [d]$ for all $i\in [n]$.
We further show that $f(i+1) \leq f(i)+1$, e.g., $f$ grows at a rate of at most $1$.
Most importantly, we provide the necessary conditions for online vertices to be unmatched using the offset function.
In particular, an online vertex left unmatched by the down-matching rule can occur only at an index whose $f$-value is strictly larger than those at the preceding $d$ indices, i.e., $v_i$ is unmatched only if
\begin{equation*}
    f(i) > f(j), \qquad \forall j\in [i-d,i-1].
\end{equation*}
Symmetrically, an online vertex left unmatched by the up-matching rule can occur only at an index whose $f$-value is strictly smaller than those at the following $d$ indices. 
We refer to these local extrema as bad indices, and let $H$ be the set of bad indices for down-matching, and $L$ be that for up-matching.
Since only vertices corresponding to bad indices can be unmatched, lower bounding the competitive ratio reduces to upper bounding the total number of bad indices:
\begin{equation}
    |H|+|L|\leq \frac{2d-2}{3d-2}\cdot n.
    \label{equation:H+L_at_most_2n/3}
\end{equation}

\paragraph{Charging Argument for Bounding the Bad Indices.}
Intuitively, there cannot be too many bad indices. If $i$ is an $H$-bad index, then $f(i)$ is larger than the values at all $d$ preceding indices.
If one of these preceding indices is also $H$-bad, then its value must in turn exceed the values at its own $d$ preceding indices. 
Consequently, along any sequence of $H$-bad indices in which consecutive indices are at distance at most $d$, the corresponding function values are strictly increasing.
We call such a sequence a \emph{chain} of bad indices.
We show that every chain has length at most $d-1$ and that any two maximal chains are separated by at least $d$ indices. This immediately gives the upper bound $\frac{d-1}{2d-1}\cdot n$ on the number of indices in $H$.
The same upper bound holds for $L$.
Although each bound can be tight individually, as witnessed by Figure~\ref{figure:up-down-hard-instances}, a strictly better bound holds for $|H|+|L|$. 
However, this relies on a careful handling of the interactions between $H$-chains and $L$-chains.
The key interaction lemma shows that when an $H$-chain of length $x$ and an $L$-chain of length $y$ are within a distance of at most $2d$, then their total length is at most $d-1$.
We convert this local restriction into a global charging argument by assigning a proxy of indices to every maximal chain.
The length of a proxy of a chain is proportional to the length of the chain.
Proxies associated with chains of the same type are disjoint.
The proxies are constructed so that an intersection between two proxies of opposite types implies $x+y\leq d-1$, which allows us to pack all proxies within a space of size $n$.
By upper bounding the total length of proxies and using the symmetry between $H$ and $L$, we prove Equation~\eqref{equation:H+L_at_most_2n/3}, which then yields a competitive ratio of $\frac{2d-1}{3d-2}$.

% \paragraph{The $1-1/e$ Upper Bound.}
% Our construction of the hard instance is recursive.
% At level $k$, consider an interval $I$ containing $k!$ offline vertices and partition it into $k$ consecutive subintervals of size $(k-1)!$.
% We introduce $(k-1)!$ online vertices whose common neighborhood is $I$.
% One of the $k$ subintervals is then selected randomly to have no future neighbors, while the construction continues in parallel within each of the other $k-1$ subintervals.
% All random choices are sampled before the matching begins, the resulting distribution is thus oblivious.

% This recursive structure forces the algorithm to make decisions without knowing which part of the instance will become ``inactive" in the future, and this uncertainty propagates across all levels.
% By symmetry, we may assume w.l.o.g. that the algorithm treats all vertices that arrive at the same batch uniformly.
% This reduces the process to a recursive allocation problem, where online vertices are distributed across $k$ identical options, one of which is later removed.
% Solving the resulting recursion and letting $k\to \infty$ yields the $1-1/e$ bound.

\subsection{Other Related Work}

Due to the vast literature on online bipartite matching, we review only the most relevant studies here.
For broader surveys, we refer to Mehta~\cite{journals/fttcs/Mehta13} and the more recent survey of Huang et al.~\cite{journals/sigecom/HuangTW24}.

\paragraph{Online Matching with Bounded Degree.}
A major line of work obtains better guarantees for online bipartite matching by exploiting the graph's degree structure.
Buchbinder et al.~\cite{conf/esa/BuchbinderJN07} considered graphs with maximum online degree $d$ and proposed a deterministic fractional algorithm with competitive ratio $1-(1-1/d)^d$.
Azar et al.~\cite{conf/soda/AzarCR17} proved a matching impossibility result for deterministic fractional algorithms.
Naor and Wajc~\cite{journals/teco/NaorW18} introduced $(k,d)$-bounded graphs in which offline vertices have degrees at least $k$ while online vertices have degrees at most $d$ with $k\geq d\geq 2$, and gave an optimal deterministic algorithm with a competitive ratio of $1-(1-1/d)^k$.
For $d$-regular graphs, Cohen and Wajc~\cite{conf/soda/CohenW18} designed a randomized algorithm whose competitive ratio tends to $1$ as $d$ grows.
Note that our uniform-length setting imposes no
nontrivial bound on the offline degrees; it is therefore incomparable with both $(k,d)$-boundedness and regularity.
Further work has studied different degree regimes, vertex capacities, and related allocation problems~\cite{conf/esa/AlbersS22,conf/wine/AlbersS22,conf/esa/CohenP23,conf/wine/FengLWZ25}.
More recently, Bhangale et al.~\cite{conf/isaac/BhangaleCH25} determined the optimal randomized competitive ratio of approximately $0.7178$ when each online vertex has degree at most $2$.

\paragraph{Convex and Geometric Matching.}
Glover~\cite{journals/Glover1967} introduced convex bipartite graphs and showed that a maximum matching can be obtained by a greedy rule.
Lipski and Preparata~\cite{journals/acta/LipskiP81} gave near-linear and linear-time algorithms for convex and doubly convex graphs, respectively.
Subsequent works~\cite{journals/Gallo1984,journals/Scutella1988} further reduced the time complexity, resulting in an $O(|V|)$-time algorithm by Steiner and Yeomans~\cite{journals/Steiner1996}.
Katriel~\cite{journals/informs/Katriel08} studied the vertex-weighted variant of convex bipartite matching and related it to scheduling unit-length jobs with release times and deadlines.
Brodal et al.~\cite{conf/mfcs/BrodalGHK07} studied the dynamic maintenance of maximum matchings in convex bipartite graphs.
Abels et al.~\cite{journals/acta/AbelsAB26} studied an online-over-time variant on interval-constrained bipartite graphs, allowing previously matched online vertices to be reassigned to later time slots.
Sentenac et al.~\cite{journals/mor/SentenacNLMP26} considered one-dimensional bipartite random geometric graphs; these graphs are doubly convex when both sides are ordered by location.
Unlike the former model, ours permits no reassignment; unlike the latter, we seek worst-case guarantees under adversarial graphs and arrival orders.

\paragraph{Max-min Greedy Matching.}
Analyzing \textsc{Flip} against the semi-adaptive adversary is closely related to the max-min greedy matching problem introduced by Eden et al.~\cite{journals/toc/EdenFF22}.
Given a bipartite graph with a perfect matching of size $n$, a max player first chooses a priority order on the offline vertices; a min player then chooses a worst-case arrival order of the online side.
Eden et al.~\cite{journals/toc/EdenFF22} showed that the max player can always guarantee a matching size of at least $0.51\cdot n$.
Importantly, the priority order chosen by the max player may depend on the graph.
Eden et al.~\cite{journals/toc/EdenFF22} showed that a uniformly random priority order gets an expected matching size of only $(0.5+o(1))\cdot n$ in the worst case.
In our uniform-length setting, \textsc{Flip} and the semi-adaptive adversary follow the same order of play. However, \textsc{Flip} randomizes only between two graph-independent priority orders: the increasing and decreasing orders of the offline vertices.
Nevertheless, we show that the average of their worst-case matching sizes is at least $2/3\cdot n$ whenever each online vertex is adjacent to exactly $d$ consecutive offline vertices.
In contrast to general bipartite graphs, the convex structure on the offline side therefore makes it possible to obtain a guarantee above $0.5$ using priority orders that are independent of the graph.
On the negative side, Cohen-Addad et al.~\cite{conf/sigecom/Cohen-AddadEFF16} observed that there exists a graph for which every priority order has a worst-case matching size of at most $2/3\cdot n$.
Chan et al.~\cite{journals/tcs/ChanTX24} established strong conditional hardness results for computing an optimal response of the min player to a given priority order.

\section{Preliminaries}

\paragraph{Notations.}
We write $[k]$ for $\{1,2,\ldots,k\}$ and $[a,b]$ for $\{a,a+1,\ldots,b\}$.

Let $G=(U\cup V,E)$ be a bipartite graph, where $U=\{u_1,\ldots,u_m\}$ is the set of offline vertices and $V=\{v_1,\ldots,v_n\}$ is the set of online vertices.
The offline vertices, together with their order $u_1,u_2,\ldots,u_m$, are known to the algorithm in advance.
The online vertices arrive one at a time in an adversarial order.
Upon the arrival of an online vertex $v$, the algorithm observes its neighborhood $N(v)\subseteq U$ and must irrevocably either match $v$ to an unmatched vertex in $N(v)$ or leave $v$ unmatched.

We impose the following two structural assumptions.
\begin{itemize}
    \item \textbf{Interval neighborhoods.} For every online vertex $v$, its neighborhood $N(v)$ forms an interval on the ordered offline set $U$.
    In other words, if offline vertices $u_i,u_j\in N(v)$ for some $i<j$, then we also have $u_k\in N(v)$ for every $k\in [i,j]$.
    
    \item \textbf{Uniform-length.} Every online vertex has exactly $d$ neighbors, where $d\geq 2$.
    Thus, for each $v\in V$, the set $N(v)$ consists of $d$ consecutive offline vertices.
\end{itemize}

For an instance $\mathcal I$ of the problem, let $\alg(\mathcal I)$ denote the size of the matching produced by an online algorithm, and let $\opt(\mathcal I)$ denote the size of a maximum matching.
The competitive ratio of a possibly randomized online algorithm is 
\begin{equation*}
    \inf_{\mathcal I} \left\{\frac{\mathbb E[\alg(\mathcal I)]}{\opt(\mathcal I)} \right\},
\end{equation*}
where the expectation is taken over the internal randomness of the algorithm.

When the uniform-length assumption is dropped, so that $|N(v)|$ may vary, we establish an upper bound of $1-1/e$ on the competitive ratio of every online algorithm.
The proof is given in Section~\ref{section:1-1/e}.

We index the online vertices so that their neighborhoods are ordered from top to bottom.
More precisely, for each $i\in[n]$, let  
\begin{equation*}
    s_i:=\min\{j:u_j\in N(v_i)\} \qquad\text{and}\qquad l_i:=\max\{j:u_j\in N(v_i)\},
\end{equation*}
and relabel the online vertices so that $s_1\leq s_2\leq\cdots\leq s_n$, breaking ties arbitrarily.
Since the neighborhood of each online vertex has length $d$, we have $l_i=s_i+d-1$ for every $i\in [n]$.
Consequently, $l_1\leq l_2\leq\cdots\leq l_n$.
These indices do not necessarily reflect the online arrival order.

\section{A Simple Randomized Algorithm: \textsc{Flip}}

We now introduce a randomized algorithm, called \textsc{Flip}, that uses only one random bit.
The algorithm chooses uniformly at random between the two natural priority orders on the offline vertices.

\begin{tcolorbox}[title=\textsc{Flip}]
    Sample $\text{flip}\in \{0,1\}$ uniformly at random.

    When an online vertex $v$ arrives:
    \begin{itemize}
        \item if $\text{flip}=0$, match it to its highest (smallest-indexed) unmatched neighbor;
        \item if $\text{flip}=1$, match it to its lowest (largest-indexed) unmatched neighbor.
    \end{itemize}

    If $v$ has no unmatched neighbor, leave it unmatched.
\end{tcolorbox}

We call the deterministic rule used when $\text{flip}=0$ the \emph{up-matching rule}, and the rule used when $\text{flip}=1$ the \emph{down-matching rule}.
For a graph $G$ and an arrival order $\sigma$, let $\alg_{\mathrm{up}}(G,\sigma)$ denote the size of the matching produced by the up-matching rule, and let $\alg_{\mathrm{down}}(G,\sigma)$ denote the corresponding quantity for the down-matching rule.
Note that both $\alg_{\mathrm{up}}(G,\sigma)$ and $\alg_{\mathrm{down}}(G,\sigma)$ are deterministic, and we have
\begin{equation*}
    \E[\alg] = \frac{1}{2}\cdot \Big( \alg_{\mathrm{up}}(G,\sigma) + \alg_{\mathrm{down}}(G,\sigma) \Big).
\end{equation*}

\subsection{A Semi-adaptive Adversary and Worst-case Arrival Orders}

We make the standard assumption of an oblivious adversary.
That is, the adversary can decide the graph and the arrival order, but cannot see the internal randomness, i.e., the flip variable, when designing the instance.
% In fact, it can be easily shown (using hard instances similar to the one we show in Section~) that 
Recall that Figure~\ref{figure:up-down-hard-instances} gives simple instances on which either up-matching or down-matching matches only about half of the vertices.
Therefore, \textsc{Flip} cannot achieve a ratio better than $1/2$ against a fully adaptive adversary.
Nevertheless, we show that its ratio is at least $2/3$ when the adversary is \emph{semi-adaptive}.

Specifically, we analyze the performance of \textsc{Flip} against the following semi-adaptive adversary.
\begin{enumerate}
    \item The adversary, who knows the algorithm, chooses the graph $G$.
    
    \item The algorithm samples the variable $\text{flip}\in \{0,1\}$.
    
    \item After observing $\text{flip}$, the adversary chooses the arrival order of the online vertices.
\end{enumerate}

Therefore, the adversary may choose one arrival order $\sigma_0$ when $\text{flip}=0$ and another arrival order $\sigma_1$ when $\text{flip}=1$.
On a fixed graph $G$, the expected matching size of \textsc{Flip} is thus
\begin{equation*}
    \min_{\sigma_0}\left\{ \frac{1}{2}\cdot \alg_{\mathrm{up}}(G,\sigma_0) \right\} + \min_{\sigma_1} \left\{ \frac{1}{2}\cdot \alg_{\mathrm{down}}(G,\sigma_1) \right\}.
\end{equation*}

\begin{theorem}
\label{theorem:flip-competitive-ratio}
    Against the semi-adaptive adversary, \textsc{Flip} is $\frac{2d-1}{3d-2}$-competitive.
\end{theorem}

Before proving the theorem, we reduce the analysis of \textsc{Flip} to graphs in which some maximum matching covers every online vertex.
We then construct a maximum matching and identify a worst-case arrival order for each of the two deterministic matching rules.

Notice that for either rule, inserting additional online vertices while preserving the relative order of the original vertices cannot decrease the final matching size.
Inserting a single online vertex either has no effect or creates an alternating path.
The path ends either at a previously unmatched offline vertex, increasing the matching size by one, or at a previously matched online vertex that becomes unmatched, leaving the matching size unchanged.

Let $M$ be a maximum matching of $G$, and let $W\subseteq V$ be the set of online vertices covered by $M$.
Since $M$ is maximum, we have $\opt(G) = |W| =\opt(G_W)$.
For any online arrival order $\sigma$, let $\sigma_W$ be the subsequence obtained by deleting the vertices in $V\setminus W$.
The preceding argument gives, for each rule $r\in\{\mathrm{up},\mathrm{down}\}$, $\alg_r(G,\sigma) \geq \alg_r(G_W,\sigma_W)$, and hence
\begin{equation*}
    \min_{\sigma}\{\alg_r(G,\sigma)\} \geq \min_{\tau}\{\alg_r(G_W,\tau)\},
\end{equation*}
where the second minimum ranges over all arrival orders of $W$.
Therefore, it suffices to prove Theorem~\ref{theorem:flip-competitive-ratio} for graphs that admit a maximum matching covering all online vertices.
Henceforth, we assume without loss of generality that $\opt=|V|=n$.

Recall that the online vertices are indexed so that $s_1\leq s_2\leq\cdots\leq s_n$.
Since all interval neighborhoods have the same length, this also gives $l_1\leq l_2\leq\cdots\leq l_n$.
Consider running up-matching with arrival order
$v_1,v_2,\ldots,v_n$.
Let $M^\star$ denote the resulting matching.
This is precisely the greedy procedure for convex bipartite matching studied by Glover~\cite{journals/Glover1967}, which yields the following lemma.
For completeness, we provide an alternative proof as follows.

\begin{lemma}
\label{lemma:canonical-maximum-matching}
    $M^\star$ is a maximum matching of $G$.
\end{lemma}
\begin{proof}
    For a matching $M$ and an online vertex $v$, we write $M(v)=u$ if $(u,v)\in M$, and $M(v)=\bot$ if $v$ is unmatched by $M$.
    Let there be a dummy online vertex $v_0$ that has no offline neighbor.
    We show inductively that, for every $i\in[0,n]$, there exists a maximum matching $M$ such that for any $j\in [0,i]$, we have $M(v_j)=M^\star(v_j)$.
    Suppose the claim holds for $i-1$, and consider the following cases:
    \begin{itemize}
        \item $M^\star(v_i)=\bot$.
        Then every neighbor of $v_i$ is occupied by an edge of $M^\star$ incident to a smaller-indexed online vertex.
        These edges also belong to $M$, and $M$ must leave $v_i$ unmatched as well.

        \item $M^\star(v_i)=u_a$, while $M(v_i)=\bot$.
        Because $M$ is maximum, $u_a$ must be matched in $M$, say to $v_k$.
        Since $u_a$ is not matched to any of $v_1,\ldots,v_{i-1}$, we have $k>i$.
        Replacing $(v_k,u_a)$ by $(v_i,u_a)$ preserves the size of the matching and makes $M(v_i)=u_a$.

        \item $M^\star(v_i)=u_a$ while $M(v_i)=u_b$, where $a\neq b$.
        The vertex $u_b$ is not matched to any of $v_1,\ldots,v_{i-1}$. Since $M^\star$ greedily chooses the smallest-indexed available neighbor of $v_i$, we have $a<b$.
        If $u_a$ is unmatched in $M$, we may simply replace $(v_i,u_b)$ by $(v_i,u_a)$.
        Otherwise, let $v_k$ (with $k>i$) be matched to $u_a$.
        Since $u_a\in N(v_k)$, we have $s_k\leq a$.
        Together with $a<b$ and $l_i\leq l_k$, we have
        \begin{equation*}
            s_k\leq a < b\leq l_i\leq l_k,
        \end{equation*}
        and hence $u_b\in N(v_k)$.
        We may therefore replace $\{(v_i,u_b), (v_k,u_a)\}$ by $\{(v_i,u_a),(v_k,u_b)\}$.
    \end{itemize}
    
    Therefore, in every case, we obtain a maximum matching $M$ that satisfies $M(v_i)=M^\star(v_i)$.
    The lemma follows by induction.
\end{proof}

By the preceding reduction, $M^\star$ covers every online vertex.
By symmetry, the down-matching rule is optimal under the reverse order $v_n,v_{n-1},\ldots,v_1$.
The next lemma shows that each of these two orders is a worst-case order for the opposite rule (see Figure~\ref{figure:up-down-hard-instances} for an example).

\begin{lemma}[Worst-case Arrival Orders]
\label{lemma:worst-arrival-order}
    For every graph $G$, we have
    \begin{equation*}
        \min_{\sigma} \left\{ \alg_{\mathrm{up}}(G,\sigma) \right\} = \alg_{\mathrm{up}} \left(G,(v_n,v_{n-1},\ldots,v_1)\right),
    \end{equation*}
    and
    \begin{equation*}
        \min_{\sigma} \left\{\alg_{\mathrm{down}}(G,\sigma)\right\} = \alg_{\mathrm{down}} \left(G,(v_1,v_2,\ldots,v_n)\right).
    \end{equation*}
\end{lemma}
\begin{proof}
    We prove the statement for the up-matching rule.
    The proof for the down-matching rule is symmetric.
    Let $\sigma_0 = (v_n,\ldots,v_1)$.
    Consider an arbitrary arrival order $\sigma \ne \sigma_0$.
    Then $\sigma$ contains two consecutively arriving vertices $v_i$ and $v_j$ with $i<j$.
    Let $\sigma'$ be the arrival order obtained by swapping these two consecutive vertices in $\sigma$, so that $v_j$ arrives immediately before $v_i$.
    We show that this swap cannot increase the final matching size.
    More precisely, we have $\alg_{\mathrm{up}}(G,\sigma) \geq \alg_{\mathrm{up}}(G,\sigma')$.

    Since $i<j$, we have $l_i\leq l_j$.
    For $v_j$, neighbors in $N(v_i)\cap N(v_j)$ (if there is any) have higher priorities than those in $N(v_j)\setminus N(v_i)$.
    Immediately before $v_i$ and $v_j$ arrive:
    \begin{itemize}
        \item If no vertex in $N(v_i)\cap N(v_j)$ is unmatched, or if at least one vertex in $N(v_i)\setminus N(v_j)$ is unmatched, the two resulting matchings under $\sigma$ and $\sigma'$ are identical.
        
        \item If at least two vertices in $N(v_i)\cap N(v_j)$ are unmatched, the resulting matchings differ only at $v_i$ and $v_j$, and we still have $\alg_{\mathrm{up}}(G,\sigma) = \alg_{\mathrm{up}}(G,\sigma')$.
    \end{itemize}
    
    The only remaining case is that all offline vertices in $N(v_i)$ are already matched before $v_i$ and $v_j$ arrive, except a unique vertex $w_0\in N(v_i)\cap N(v_j)$.
    In this case:
    \begin{itemize}
        \item Under order $\sigma$, $v_i$ arrives earlier and matches to $w_0$, $v_j$ matches to the smallest-indexed unmatched offline vertex $w_1$ in $N(v_j)\setminus N(v_i)$, if such $w_1$ exists;
            
        \item Under order $\sigma'$, $v_j$ arrives earlier and matches to $w_0$, $v_i$ is unmatched.
    \end{itemize}
    A standard alternating path argument gives $\alg_{\mathrm{up}}(G,\sigma) \geq \alg_{\mathrm{up}}(G,\sigma')$.
    Repeatedly applying the above adjacent swaps transforms any arrival order into $(v_n,v_{n-1},\ldots,v_1)$ without increasing the matching size.
    Therefore, $(v_n,v_{n-1},\ldots,v_1)$ is a worst-case arrival order for the up-matching rule.
    
    By reversing the order of the offline indices, the same argument shows that $(v_1,v_2,\ldots,v_n)$ is a worst-case arrival order for the down-matching rule.
\end{proof}

Having identified the worst arrival order for each rule, we next bound the total number of online vertices left unmatched by the two rules.

\subsection{From Unmatched Vertices to Bad Indices}

In this subsection, we introduce a discrete function $f$, called the \emph{offset function}, and relate the unmatched online vertices to certain one-sided local extrema of $f$, which we call \emph{bad indices}.

For each $i\in [n]$, let $\mu_i\in [m]$ be the index of the offline vertex that gets matched to $v_i$ by $M^\star$.
Recall that when $v_i$ is processed during the construction of $M^\star$, every offline vertex with index in $[s_i,\mu_{i-1}]$ is already occupied whenever this interval is non-empty, while the vertex with index $\mu_{i-1}+1$ is available.
Therefore, we have $\mu_1=s_1$, and $\mu_i=\max\{s_i,\mu_{i-1}+1\}$ for every $i\in [2,n]$.

\begin{definition}[Offset Function]
    For each online vertex $v_i$, where $i\in [n]$, define
    \begin{equation*}
        f(i) := \mu_i-s_i+1.
    \end{equation*}
    In other words, $f(i)$ is the number of neighbors of $v_i$ whose index is at most $\mu_i$; see Figure~\ref{figure:example-G-and-f}.
\end{definition}

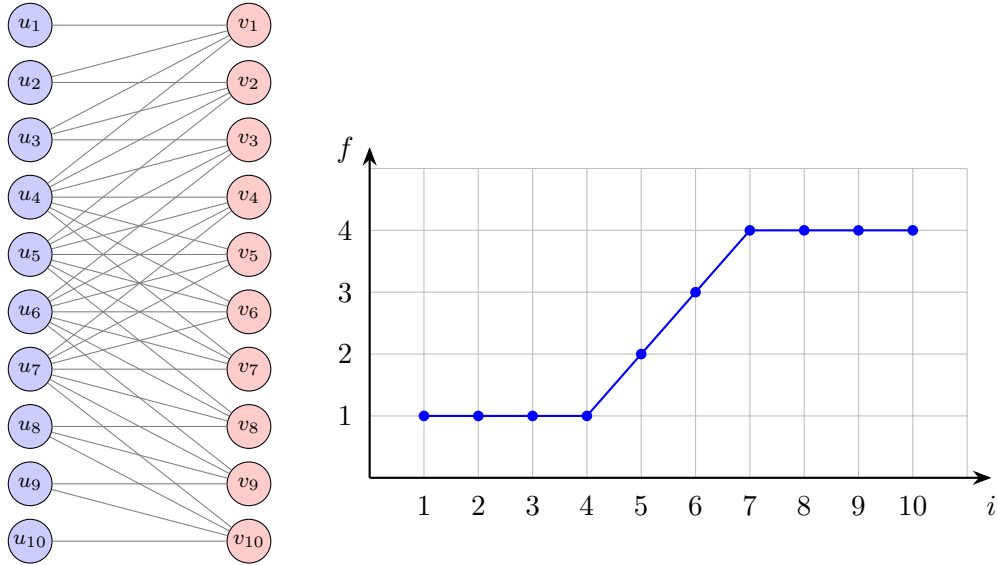
\begin{figure}[htb]
\centering
\resizebox{0.8\textwidth}{!}{%
\begin{tikzpicture}[
    vertex/.style={
        circle,
        draw,
        minimum size=5.8mm,
        text width=4.7mm,
        align=center,
        inner sep=0pt,
        font=\scriptsize
    },
    u_vertex/.style={vertex, fill=blue!20},
    v_vertex/.style={vertex, fill=red!20},
    edge/.style={draw, gray, thin}
]

% -----------------------------------------------------------------
% Left panel: graph G.
% The smaller horizontal and vertical spacings make this panel
% more compact without changing the vertex sizes.
% -----------------------------------------------------------------
\foreach \i in {1,...,10} {
    \pgfmathsetmacro{\yy}{7.2-0.76*(\i-1)}
    \node[u_vertex] (lu\i) at (0,\yy) {$u_{\i}$};
    \node[v_vertex] (lv\i) at (2.9,\yy) {$v_{\i}$};
}

\foreach \j in {1,...,4} {
    \draw[edge] (lv1) -- (lu\j);
}
\foreach \j in {2,...,5} {
    \draw[edge] (lv2) -- (lu\j);
}
\foreach \j in {3,...,6} {
    \draw[edge] (lv3) -- (lu\j);
}
\foreach \i in {4,...,7} {
    \foreach \j in {4,...,7} {
        \draw[edge] (lv\i) -- (lu\j);
    }
}
\foreach \j in {5,...,8} {
    \draw[edge] (lv8) -- (lu\j);
}
\foreach \j in {6,...,9} {
    \draw[edge] (lv9) -- (lu\j);
}
\foreach \j in {7,...,10} {
    \draw[edge] (lv10) -- (lu\j);
}

% -----------------------------------------------------------------
% Right panel: offset function f.
% The y-coordinate unit is slightly smaller than before, making
% the plot slightly shorter while retaining its width.
% -----------------------------------------------------------------
\begin{scope}[shift={(4.5,1.2)}, x=0.72cm, y=0.82cm]

    % Grid: integer coordinates correspond exactly to (i,f(i))
    \draw[gray!50, thin, step=1]
        (0,0) grid (11,5);

    % x-axis
    \draw[black, thick, -{Stealth[length=2.5mm]}] (0,0) -- (11.45,0);
    \node at (11.45,-0.45) {$i$};
    \foreach \i in {1,...,10} {
        \node at (\i,-0.45) {$\i$};
    }

    % y-axis
    \draw[black, thick, -{Stealth[length=2.5mm]}] (0,0) -- (0,5.35);
    \node at (-0.45,5.3) {$f$};
    \foreach \j in {1,...,4} {
        \node at (-0.45,\j) {$\j$};
    }

    % Correct values: 1,1,1,1,2,3,4,4,4,4
    \foreach \x/\y in {
        1/1,
        2/1,
        3/1,
        4/1,
        5/2,
        6/3,
        7/4,
        8/4,
        9/4,
        10/4
    }{
        \fill[blue] (\x,\y) circle (2pt);
    }

    \draw[blue, thick]
        (1,1) -- (2,1) -- (3,1) -- (4,1)
        -- (5,2) -- (6,3) -- (7,4)
        -- (8,4) -- (9,4) -- (10,4);

\end{scope}

\end{tikzpicture}%
}
\caption{An example of graph $G$ in which $\mu_i=i$ for all
$i\in[n]$ and the corresponding offset function $f$.}
\label{figure:example-G-and-f}
\end{figure}

Since $s_i\leq \mu_i\leq l_i=s_i+d-1$, it follows immediately that $f(i)\in[d]$.
The following lemma shows the one-sided Lipschitzness of $f$.

\begin{lemma}[$f$ is one-sided Lipschitz] \label{lemma:lipschitz_of_f}
    For every $1\leq i<j\leq n$, we have $f(j)-f(i) \leq j-i$.
\end{lemma}
\begin{proof}
    The recurrence $\mu_i = \max\{s_i,\mu_{i-1}+1\}$ implies that, for every $i\in [2,n]$,
    \begin{equation*}
        f(i) = \max\{s_i,\mu_{i-1}+1\}-s_i+1
        = \max\{0, f(i-1)-(s_i-s_{i-1})\} + 1
        \leq f(i-1)+1,
    \end{equation*}
    where the inequality follows from $s_i\geq s_{i-1}$.
    Consequently, for every $i<j$,
    \begin{equation*}
        f(j)-f(i) = \sum_{k=i+1}^{j} (f(k)-f(k-1)) \leq j-i,
    \end{equation*}
    and the lemma follows.
\end{proof}

The Lipschitz property is one-sided: it limits how quickly $f$ can increase but imposes no similar restriction on how quickly it can decrease.
For example, take two disjoint copies of the graph $G$ shown in Figure~\ref{figure:example-G-and-f} and concatenate them in both the offline and online orders.
At the boundary between the two copies, the offset function drops sharply from $f(10)=4$ to $f(11)=1$ in a single step.

Since each online vertex has $d$ neighbors, none of the first $d$ arrivals would be left unmatched by either of the two matching rules.
The following two lemmas further establish the necessary conditions for an online vertex to be unmatched.

\begin{lemma}
\label{lemma:unmatched-only-if-local-max}
    Consider any online vertex $v_i$ with $i\in [d+1,n]$.
    If it is not matched by the down-matching rule under the worst arrival order, then we must have
    \begin{equation*}
        f(i) > f(j), \qquad \forall j\in [i-d,i-1].
    \end{equation*}
\end{lemma}
\begin{proof}
    Assume the contrary that there exists some $j\in [i-d, i-1]$ such that $f(j) \geq f(i)$.
    We show that $v_i$ must have an available neighbor upon arrival, contradicting the assumption that it is unmatched.
    Recall that $l_i$ (resp. $l_j$) is the largest index of a neighbor of $v_i$ (resp. $v_j$). Thus
    \begin{equation*}
        l_j = \mu_j + d - f(j), \qquad l_i = \mu_i + d - f(i).
    \end{equation*}
    Since $f(j)\geq f(i)$, we have
    \begin{equation*}
        l_i - l_j = (\mu_i-\mu_j) + (f(j) - f(i)) \geq \mu_i - \mu_j \geq i - j.
    \end{equation*}
 
    Immediately after $v_j$ arrives, every offline vertex with index larger than $l_j$ is unmatched, because no previously arrived online vertex is adjacent to it.
    
    Next, we count the number of available neighbors of $v_i$ upon its arrival.
    \begin{itemize}
        \item If $s_i > l_j$, all matched neighbors of $v_i$ are matched by vertices that arrived between $v_j$ and $v_i$.
        Note that there are only $i-j-1 \leq d-1$ such vertices, which implies that at least one neighbor of $v_i$ is not matched.

        \item If $s_i \leq l_j$, then among neighbors with index $[l_j+1, l_i]$, at most $i-j-1$ of them are matched, which implies that at least one of them is not matched, because $l_i - l_j \geq i - j$.
    \end{itemize}
    
    Thus, at least one neighbor of $v_i$ is available when $v_i$ arrives, a contradiction.
\end{proof}

By symmetry, if $v_i$ is unmatched by the up-matching rule, then $f(i)$ is strictly smaller than the values at the next $d$ indices.
The proof is omitted, as it is identical to that of Lemma~\ref{lemma:unmatched-only-if-local-max} after reversing the order.

\begin{lemma}
\label{lemma:unmatched-only-if-local-min}
    Consider any online vertex $v_i$ with $i\in [1,n-d]$.
    If it is not matched by the up-matching rule under the worst arrival order, then we must have
    \begin{equation*}
        f(i) < f(j), \qquad \forall j\in [i+1,i+d].
    \end{equation*}
\end{lemma}

Lemma~\ref{lemma:unmatched-only-if-local-max} and~\ref{lemma:unmatched-only-if-local-min} identify two types of indices at which an unmatched online vertex may occur.
We collect these indices in the following sets.

\begin{definition}[Bad Indices]
    Define
    \begin{align*}
        H &:= \left\{ x\in[d+1,n]: f(x)>\max_{y\in[x-d,x-1]} \{f(y)\} \right\}, \\
        L &:= \left\{ x\in[1,n-d]: f(x)<\min_{y\in[x+1,x+d]} \{f(y)\} \right\}.
    \end{align*}
    An index in $H$ (resp. $L$) is called an $H$-bad (resp. $L$-bad) index.
\end{definition}

Notice that an index may belong to both $H$ and $L$.
For example, in Figure~\ref{figure:example-G-and-f}, we have $H=\{5,6,7\}$ and $L=\{4,5,6\}$.

By Lemmas~\ref{lemma:unmatched-only-if-local-max} and \ref{lemma:unmatched-only-if-local-min}, every online vertex left unmatched by the down-matching rule has an $H$-bad index, whereas every online vertex left unmatched by the up-matching rule has an $L$-bad index.
Equivalently,
\begin{align*}
    & \{i\in [n]: v_i \text{ is unmatched by the down-matching rule}\} \subseteq H, \\
    & \{i\in [n]: v_i \text{ is unmatched by the up-matching rule}\} \subseteq L.
\end{align*}
Consequently, the expected matching size of \textsc{Flip} is lower-bounded by
\begin{equation*}
    \frac{1}{2} \left( (n-|H|) + (n-|L|) \right) \geq n - \frac{|H|+|L|}{2}.
\end{equation*}
Recall that $\opt = n$.
To prove Theorem~\ref{theorem:flip-competitive-ratio}, it suffices to establish the bound
\begin{equation}
\label{equation:main-inequality}
    |H|+|L|\leq \frac{2d-2}{3d-2}\cdot n.
\end{equation}

\section{Bounding the Number of Bad Indices}
\label{section:bounding-H-L}

In this section, we bound $|H|+|L|$, the total number of bad indices, where an index that is both $H$-bad and $L$-bad is counted twice.
Throughout, fix an arbitrary offset function $f:[n]\to[d]$ satisfying the one-sided Lipschitz condition, i.e., $f(i+1)\leq f(i)+1$ for every $i\in[n-1]$.
These are the only properties of $f$ used in this section; in particular, we do not require that $f$ arises from a bipartite graph with uniform interval neighborhoods.

\subsection{Maximal Chains}

We first group the bad indices into maximal chains that are well separated.

\begin{definition}[Maximal Chain]
\label{definition:maximal-chain}
    List the $H$-bad indices in increasing order.
    An \emph{$H$-chain} is a sequence $C=(i_1,i_2,\ldots,i_x)$
    of consecutive indices in this list such that $i_{k+1}-i_k\leq d$ for every $k<x$.
    The chain is \emph{maximal} if it cannot be extended by another $H$-bad index on either side.
    Let $\mathcal{C}_H$ be the collection of maximal $H$-chains.
    We define $L$-chains and the collection $\mathcal{C}_L$ analogously.
\end{definition}

For any maximal chain $C=(i_1,i_2,\ldots,i_x) \in \mathcal{C}_H\cup \mathcal{C}_L$, we denote its starting index, ending index, and length by
\begin{equation*}
    \st(C) = i_1,
    \qquad
    \ed(C) = i_x,
    \qquad
    \len(C) = x.
\end{equation*}

The maximal chains partition $H$ and $L$. 
In particular, we have
\begin{equation*}
    \sum_{C\in\mathcal{C}_H} \len(C) = |H|,
    \qquad
    \sum_{D\in\mathcal{C}_L} \len(D) = |L|.
\end{equation*}

Furthermore, maximal chains of the same type are separated by at least $d+1$ indices. Specifically, for any $C_1,C_2\in\mathcal{C}_H$ with $\st(C_1)<\st(C_2)$, we have
\begin{equation*}
    \st(C_2) \geq \ed(C_1) + d + 1.
\end{equation*}
In other words, between the last index of $C_1$ and the first index of $C_2$, there are at least $d$ other indices.
We next record a basic property of maximal chains.

\begin{lemma}
\label{lemma:chain-properties}
    Let $C=(i_1,\ldots,i_x)$ be any maximal chain.
    Then we have
    \begin{equation}
    \label{equation:chain-weight-upper-bound}
        x-1 \leq f(i_x)-f(i_1) \leq \min\{i_x - i_1, d-1\}.
    \end{equation}
\end{lemma}
\begin{proof}
    By Definition~\ref{definition:maximal-chain}, for every $k<x$, we have $i_{k+1}-i_k\leq d$, which implies $f(i_k) < f(i_{k+1})$.
    Because $f$ is integer-valued, we have $f(i_{k+1})-f(i_k) \geq 1$.
    Summing over all $k\in \{1,\ldots,x-1\}$ gives $x-1\leq f(i_x)-f(i_1)$.
    The final inequality in Equation~\eqref{equation:chain-weight-upper-bound} follows from the one-sided Lipschitz property of $f$ (Lemma~\ref{lemma:lipschitz_of_f}) and from $f(\cdot)\in[d]$.
\end{proof}

\begin{corollary}
\label{corollary:chain-length-d-minus-one}
For every $C\in\mathcal{C}_H$ and $D\in\mathcal{C}_L$, we have
\begin{equation}
\label{equation:chain-endpoint-bounds}
    \len(C)\leq d-f(\st(C))+1,
    \quad \text{ and } \quad
    \len(D)\leq f(\ed(D)).
\end{equation}
Consequently, every maximal chain has length at most $d-1$.
\end{corollary}
\begin{proof}
For every $C\in \mathcal{C}_H$, Lemma~\ref{lemma:chain-properties} implies
\begin{equation*}
    \len(C)
    \leq 1 + f(\ed(C))-f(\st(C))
    \leq 1 + d - f(\st(C)) \leq d-1,
\end{equation*}
where the last inequality holds because $\st(C)\in H$ implies $f(\st(C))\geq 2$.

Similarly, for every $D\in \mathcal{C}_L$, Lemma~\ref{lemma:chain-properties} implies
\begin{equation*}
    \len(D)
    \leq 1 + f(\ed(D))-f(\st(D))
    \leq f(\ed(D)) \leq d-1,
\end{equation*}
where the last inequality holds because $\ed(D)\in L$ implies $f(\ed(D))\leq d-1$.
\end{proof}

We next show that if an $L$-chain precedes an $H$-chain and the two chains are sufficiently close, then their total length is bounded.

\begin{lemma}
\label{lemma:chain-interaction}
Let $C\in \mathcal{C}_H$ and $D\in \mathcal{C}_L$.
If $2 \leq \st(C)-\ed(D)\leq 2d$,
then $\len(C) + \len(D) \leq d - 1$.
\end{lemma}
\begin{proof}
    By the assumption of the lemma,
    \begin{equation*}
        [\ed(D)+1,\ed(D)+d]\cap[\st(C)-d,\st(C)-1] \ne \varnothing.
    \end{equation*}
    For any $z$ in this intersection, the definitions of $L$-bad and $H$-bad indices give
    \begin{equation*}
        f(\ed(D)) < f(z) < f(\st(C)).
    \end{equation*}
    Since $f$ is integer-valued, $f(\st(C)) \geq f(\ed(D))+2$.
    Therefore, by Lemma~\ref{lemma:chain-properties},
    \begin{align*}
        \len(C) + \len(D)
        &\leq \bigl(1+f(\ed(C))-f(\st(C))\bigr)
        +\bigl(1+f(\ed(D))-f(\st(D))\bigr) \\
        &\leq d+1+f(\ed(D))-f(\st(C)) \\
        &\leq d-1,
    \end{align*}
    where the second inequality uses $f(\ed(C))\leq d$ and $f(\st(D))\geq 1$.
\end{proof}

We use these properties of maximal chains to bound $|H|+|L|$ in the following subsections.

\subsection{Proxies: A Warm-up Analysis}

As a warm-up, we first prove $|H|\leq \frac{d-1}{2d-1}\cdot n$; the symmetric argument gives the same bound for $|L|$.

Recall that $|H| = \sum_{C\in \mathcal{C}_H} \len(C)$.
We upper bound $|H|$ by the following charging argument.
For each maximal chain $C\in \mathcal{C}_H$, we define a \emph{proxy} $P(C)$, which is a subset of indices in $[n]$.
We show that 
\begin{itemize}
    \item[(1)] for each $C\in \mathcal{C}_H$, we have $|P(C)| \geq \frac{2d-1}{d-1}\cdot \len(C)$;
    \item[(2)] the proxies $\{P(C)\}_{C\in \mathcal{C}_H}$ are pairwise disjoint.
\end{itemize}

Combining these two properties, we get
\begin{equation*}
    \frac{2d-1}{d-1}\cdot |H| = \sum_{C\in \mathcal{C}_H} \frac{2d-1}{d-1} \cdot \len(C) \leq \sum_{C\in \mathcal{C}_H} |P(C)| \leq n,
\end{equation*}
where the last inequality follows since the proxies are pairwise disjoint.

It remains to construct these proxies.
For each $C\in \mathcal{C}_H$, we let
\begin{equation*}
    P(C) = [\st(C) - d, \ed(C)].
\end{equation*}

The proxy satisfies property (1) because
\begin{equation*}
    |P(C)| = d + 1 + (\ed(C) - \st(C)) \geq d + \len(C) \geq \frac{2d-1}{d-1}\cdot \len(C),
\end{equation*}
where the second inequality follows from $\len(C)\leq d-1$.
Property (2) follows directly because any maximal $H$-chain $C'$ preceding $C$ satisfies $\ed(C') \leq \st(C)-d-1$.

Symmetrically, we can also show that $|L|\leq \frac{d-1}{2d-1}\cdot n$.
Unfortunately, combining the two bounds gives only $|H|+|L|\leq \frac{2d-2}{2d-1}\cdot n$, which yields a competitive ratio approaching $1/2$ as $d\to\infty$.

A key observation is that the two inequalities $|H|\leq \frac{d-1}{2d-1}\cdot n$ and $|L|\leq \frac{d-1}{2d-1}\cdot n$ cannot be tight simultaneously.
For example, for the hard instance for down-matching in Figure~\ref{figure:up-down-hard-instances}, we have $n=2d-1$ and $|H|=d-1$.
Hence $|H| = \frac{d-1}{2d-1}\cdot n$.
However, in this case $L=\varnothing$, which gives a strictly better bound on $|H|+|L|$.
Similarly, for the hard instance for up-matching, we have $|L| = \frac{d-1}{2d-1}\cdot n$ but $H = \varnothing$.
This suggests that we should design proxies simultaneously for maximal chains in $\mathcal{C}_H$ and $\mathcal{C}_L$, and provide an upper bound for their combined size.
Consider the naive combination of the proxies above for chains in $\mathcal{C}_H$ and $\mathcal{C}_L$:
\begin{itemize}
    \item let $P(C) = [\st(C) - d, \ed(C)]$ for each $C\in \mathcal{C}_H$;
    \item let $Q(D) = [\st(D), \ed(D) + d]$ for each $D\in \mathcal{C}_L$.
\end{itemize}

While the proxies $\{P(C)\}_{C\in \mathcal{C}_H}$ are pairwise disjoint and the proxies $\{Q(D)\}_{D\in \mathcal{C}_L}$ are pairwise disjoint, there may be \emph{overlaps} between some $P(C)$ and $Q(D)$.
If there were no overlaps, the same charging argument would give $|H|+|L|\leq \frac{d-1}{2d-1}\cdot n$, which is false in general.

Consider an $L$-chain $D$ preceding an $H$-chain $C$.
Since $Q(D)$ extends the interval of $D$ by a distance of $d$ to the right and $P(C)$ extends the interval of $C$ by a distance of $d$ to the left, if $Q(D)\cap P(C) \neq \varnothing$, then the two chains must have a distance at most $2d$, i.e., $\st(C) - \ed(D) \leq 2d$.
However, in such a case, we can use Lemma~\ref{lemma:chain-interaction} to show that the total length of $C$ and $D$ is at most $d-1$.
This interaction bound allows us to use shorter proxies than those defined above.
Intuitively, when a chain is short, say $\len(C)\leq(d-1)/2$, its proxy need not extend the chain by a full distance of $d$.
We can therefore shorten the proxies of short chains to reduce potential overlaps.
Moreover, the proxy design should distinguish chains longer than $(d-1)/2$ from those of length at most $(d-1)/2$, since two long chains cannot interact in this way.

\medskip

As a first target, consider proving $|H|+|L|\leq 2n/3$, which is equivalent to
\begin{equation*}
    3\cdot |H| + 3\cdot |L| \leq 2\cdot n.
\end{equation*}

Ideally, we want to design a proxy $P(C)$ for every $C\in \mathcal{C}_H$ and a proxy $Q(D)$ for every $D\in \mathcal{C}_L$, such that the following properties hold:
\begin{itemize}
    \item[1] for each $C \in \mathcal{C}_H$, we have $|P(C)| \geq 3\cdot \len(C)$;
    \item[2] for each $D \in \mathcal{C}_L$, we have $|Q(D)| \geq 3\cdot \len(D)$;
    \item[3] every index in $[n]$ is covered at most twice by the proxies $\{ P(C) \}_{C\in \mathcal{C}_H} \cup \{ Q(D) \}_{D\in \mathcal{C}_L}$.
\end{itemize}

If these three properties hold, the same counting argument gives $|H|+|L|\leq {2n}/{3}$.
Two obstacles prevent such a direct construction.

First, observe that we cannot require the proxies $\{P(C)\}_{C\in \mathcal{C}_H}$ to be pairwise disjoint, as otherwise we can show that $|H| \leq n/3$, which is incorrect.
Thus some overlap between the proxies is unavoidable.
In particular, overlaps may occur not only between proxies of different types, as discussed above, but also between proxies of the same type.

Second, we may no longer be able to require $P(C)$ to be an ordinary subset of $[n]$.
For an extreme example, consider the hard instance for down-matching in Figure~\ref{figure:up-down-hard-instances}. Here $n=2d-1$, and the unique maximal $H$-chain $C$ has length $d-1$.
Therefore, it is infeasible to have both $P(C) \subseteq [n]$ and $|P(C)|\geq 3\cdot \len(C)$, for any $d\geq 3$.
Note that this can only happen when $\len(C) > (d-1)/2$.
If $\len(C)\leq (d-1)/2$, then the interval $[\st(C)-d,\ed(C)]$ has size at least $d+\len(C)\geq 3\cdot\len(C)$, so it contains enough space for the proxy.

In summary, these observations suggest constructing proxies for $H$- and $L$-chains simultaneously and bounding their total size.
The construction should distinguish different chain lengths and use two tracks to accommodate the larger proxies. 
Finally, strengthening the target bound from $|H|+|L|\leq 2n/3$ to $|H|+|L|\leq \frac{2d-2}{3d-2}\cdot n$ requires one additional unit of proxy size per maximal chain.

\subsection{A Two-Track Proxy Packing}
\label{subsection:two-track-proxies}

We now present a two-track proxy packing motivated by the observations above.

Let $\kappa:=|\mathcal{C}_H|+|\mathcal{C}_L|$ be the total number of maximal chains. We prove that
\begin{equation}
\label{equation:strong-chain-packing}
    3(|H|+|L|)+\kappa\leq 2n.
\end{equation}

The proof uses two copies of $[n]$, called \emph{tracks}. On the first track, every $H$-chain receives a large proxy and every $L$-chain receives a small proxy; on the second track, these roles are reversed. We show that the proxies on each track can be packed into $[n]$ and that the two proxies assigned to each chain $C$ have total size $3\cdot\len(C)+1$. After the track labels are removed, every index of $[n]$ is covered at most twice.

Before describing the construction, we show that Equation~\eqref{equation:strong-chain-packing} implies the desired bound on $|H|+|L|$.
By Corollary~\ref{corollary:chain-length-d-minus-one}, every maximal chain has length at most $d-1$. Therefore,
\begin{equation*}
    \kappa\geq\frac{|H|+|L|}{d-1}.
\end{equation*}
Combining this inequality with \eqref{equation:strong-chain-packing} gives
\begin{equation*}
    3(|H|+|L|)+\frac{|H|+|L|}{d-1}\leq 2n.
\end{equation*}
Rearranging the inequality yields the desired upper bound:
\begin{equation*}
    |H|+|L| \leq \frac{2d-2}{3d-2}\cdot n.
\end{equation*}

\smallskip

We next construct the proxies, beginning with the first track.
For every $C\in \mathcal{C}_H$, define its proxy as
\begin{equation*}
    P(C) :=
    \begin{cases}
        [\st(C)-d,\st(C)-d+3\cdot \len(C)], & \len(C) \leq (d-1)/2,\\
        [\st(C)-d,\st(C)+\len(C)-1], & \len(C) >(d-1)/2.
    \end{cases}
\end{equation*}

We call a chain $C$ \emph{short} if $\len(C)\leq (d-1)/2$, and \emph{long} otherwise.
\begin{itemize}
    \item When $C$ is short, the proxy has size $3\cdot \len(C)+1$.
    \item When $C$ is long, the proxy has size $\len(C)+d$.
\end{itemize}
Thus $|P(C)|=p(\len(C))$, where $p(x):=\min\{3x+1,d+x\}$.
Moreover, $P(C)\subseteq[n]$: we have $\st(C)\geq d+1$, and the last index of $P(C)$ is at most $\ed(C)$.

Our goal is to assign two proxies to each chain $C$ whose combined size is $3\cdot \len(C)+1$.
For short $C$, the proxy $P(C)$ we construct above already has size $3\cdot \len(C) + 1$, which means that its proxy in the second track is empty.
For long $C$, the proxy $P(C)$ has size $\len(C) + d < 3\cdot \len(C) + 1$, which means that its proxy in the second track should have size $2\cdot \len(C) - d + 1$.
Since the roles of $H$ and $L$ are reversed on the second track, the first-track proxy for an $L$-chain should follow the same rule as the second-track proxy for an $H$-chain.

Therefore, for every $D\in\mathcal{C}_L$, we define its proxy as
\begin{equation}
\label{equation:small-L-proxy}
    Q(D):=
    \begin{cases}
        \varnothing, & \len(D) \leq (d-1)/2,\\
        [\ed(D)+d-\len(D)+1,\ed(D)+\len(D)+1], & \len(D) > (d-1)/2.
    \end{cases}
\end{equation}

Thus $|Q(D)|=q(\len(D))$, where $q(x):=\max\{0,2x-d+1\}$.
Moreover, $Q(D)\subseteq[n]$: since $\ed(D)\leq n-d$ and $\len(D)\leq d-1$, its right endpoint is at most $n$.

The key step is to show that the total size of the first-track proxies is at most $n$.

\begin{lemma}[First-track Packing]
\label{lemma:first-copy-packing}
We have
\begin{equation}
\label{equation:first-copy-packing}
    \sum_{C\in\mathcal{C}_H}p(\len(C))
    +\sum_{D\in\mathcal{C}_L}q(\len(D))
    \leq n.
\end{equation}
\end{lemma}

We defer the proof to the next subsection and now construct the second-track proxies symmetrically.
Importantly, for every $x\in[d-1]$, we have $p(x)+q(x)=3x+1$.

We define the second track proxies based on a rotated version of $f$.
Specifically, we define
\begin{equation*}
    \overleftarrow f(i):=d+1-f(n+1-i).
\end{equation*}
Geometrically, $\overleftarrow f$ is obtained by rotating $f$ by 180 degrees, reversing both the horizontal and vertical directions.
For $i\leq j$, we have
\begin{equation*}
    \overleftarrow f(j)-\overleftarrow f(i)
    =f(n+1-i)-f(n+1-j)
    \leq j-i.
\end{equation*}
Thus $\overleftarrow f$ satisfies the same one-sided Lipschitz condition.
Moreover, the map $i\mapsto n+1-i$ exchanges the $H$-chains and the
$L$-chains without changing their lengths. 
For example, every $C\in\mathcal{C}_H$ for $f$ becomes an $L$-chain for $\overleftarrow f$.
Applying the same construction to $\overleftarrow f$ therefore gives
\begin{equation}
\label{equation:second-copy-packing-swapped}
    \sum_{C\in\mathcal{C}_H}q(\len(C))
    + \sum_{D\in\mathcal{C}_L}p(\len(D))
    \leq n.
\end{equation}
Here $\mathcal{C}_H$ and $\mathcal{C}_L$ continue to denote the chains of the original function $f$.

Adding \eqref{equation:first-copy-packing} and
\eqref{equation:second-copy-packing-swapped}, and using
$p(x)+q(x)=3x+1$, gives
\begin{align*}
    \sum_{C\in\mathcal{C}_H}
    \bigl(3\cdot \len(C)+1\bigr)
    +
    \sum_{D \in\mathcal{C}_L}
    \bigl(3\cdot \len(D)+1\bigr)
    =3(|H|+|L|)+\kappa \leq 2n.
\end{align*}
This proves \eqref{equation:strong-chain-packing} and yields the following main result.

\begin{theorem}
\label{theorem:exact-bad-index-bound}
For every offset function $f:[n]\to[d]$ satisfying the one-sided Lipschitz
condition, we have
\begin{equation*}
    |H|+|L|\leq \frac{2d-2}{3d-2}\cdot n.
\end{equation*}
Consequently, \textsc{Flip} is
$\frac{2d-1}{3d-2}$-competitive against the semi-adaptive adversary.
\end{theorem}

\subsection{Upper Bounding the Total Size of Proxies}

We now prove Lemma~\ref{lemma:first-copy-packing}, which states that the total size of the first-track proxies is at most $n$.
If all proxies were disjoint, the result would be immediate.
The proxies may overlap, so the key is to control these overlaps and eliminate them without changing the total proxy size.

Our analysis has three main steps:
\begin{itemize}
    \item We first show that proxies of the same type are pairwise disjoint (see Lemma~\ref{lemma:same_type_proxy_disjoint}).
    Consequently, overlap can only happen between proxies of different types.

    \item Then we show that every proxy can overlap with at most one other proxy, which is of a different type, in Lemma~\ref{lemma:at_most_one_overlap}.
    This implies that we can group overlapping proxies into pairs.

    \item Finally, for every overlapping pair, we relocate one proxy without changing its size so that the resulting proxies are pairwise disjoint.
\end{itemize}

Recall that for any $C\in \mathcal{C}_H$, its proxy is defined as
\begin{equation*}
    P(C) =
    \begin{cases}
        [\st(C)-d,\st(C)-d+3\cdot \len(C)], & \len(C) \leq (d-1)/2,\\
        [\st(C)-d,\st(C)+\len(C)-1], & \len(C) >(d-1)/2.
    \end{cases}
\end{equation*}
Define $\tilde{P}(C) := [\st(C)-d,\st(C)+\len(C)-1]$ as the \emph{extended proxy} of $C$.
It can be verified that regardless of whether $C$ is short or long, we always have $P(C) \subseteq \tilde{P}(C)$.

Also recall that the proxy of every $D\in \mathcal{C}_L$ is defined as
\begin{equation*}
    Q(D) =
    \begin{cases}
        \varnothing, & \len(D) \leq (d-1)/2,\\
        [\ed(D)+d-\len(D)+1,\ed(D)+\len(D)+1], & \len(D) > (d-1)/2.
    \end{cases}
\end{equation*}

Define $\tilde{Q}(D):=[\st(D),\ed(D)+d]$ as the extended proxy of $D$. Then $Q(D)\subseteq\tilde{Q}(D)$.

We prove a statement stronger than disjointness of proxies of the same type: we show that even their extended versions are disjoint.

\begin{lemma} \label{lemma:same_type_proxy_disjoint}
    The extended proxies of the same type are pairwise disjoint.
\end{lemma}
\begin{proof}
    Consider any $C\in \mathcal{C}_H$.
    Note that the last index of $\tilde{P}(C)$ is
    \begin{equation*}
        \st(C) + \len(C) - 1 \leq \ed(C),
    \end{equation*}
    which is smaller than $\st(C')-d$ for the succeeding $H$-chain by maximality.
    Therefore, the extended proxies of the $H$-chains are pairwise disjoint.
    
    Now we consider any $D\in \mathcal{C}_L$.
    The last index of $\tilde{Q}(D)$ is $\ed(D) + d$, which is smaller than $\st(D')$ for any succeeding $L$-chain by maximality.
    Hence, the extended proxies of the $L$-chains are also pairwise disjoint.
\end{proof}

The lemma also implies that, for chains of the same type, the order of the extended proxies agrees with the order of the chains.
Consequently, only proxies of different types may overlap.
Next, we show that if two proxies of different types overlap, then the $H$-chain is short and the $L$-chain is long.

\begin{lemma}
\label{lemma:collision_P(C)_Q(D)}
Suppose $P(C)\cap Q(D)\ne\varnothing$, where
$C\in\mathcal{C}_H$ and $D\in\mathcal{C}_L$.
Then we have $\len(C) + \len(D) \leq d-1$.
Furthermore, $C$ is short, and $D$ is long.
\end{lemma}
\begin{proof}
Let $x:=\len(C)$ and $y:=\len(D)$.
Since $P(C)$ intersects $Q(D)$, the left endpoint of $Q(D)$ is at most the right endpoint of $P(C)$. 
Consequently,
\begin{equation*}
    \ed(D)+d-y+1 \leq \st(C)-d-1+\min\{3x+1,d+x\},
\end{equation*}
which implies
\begin{equation}
    \ed(D) - \st(C) \leq x+y-d-2.
    \label{equation:ed(D)-st(C)_leq_x+y-d-2}
\end{equation}

We first show that $\ed(D)<\st(C)$. Otherwise, by the one-sided Lipschitz property (Lemma~\ref{lemma:lipschitz_of_f}) and Corollary~\ref{corollary:chain-length-d-minus-one},
\begin{equation*}
    x + y \leq d+1+f(\ed(D))-f(\st(C))
    \leq d+1+\ed(D)-\st(C),
\end{equation*}
contradicting Equation~\eqref{equation:ed(D)-st(C)_leq_x+y-d-2}.
Therefore $\st(C)>\ed(D)$.
We further rule out $\st(C)-\ed(D)=1$. In this case, since $\ed(D)\in L$ and $\st(C)=\ed(D)+1$, the definition of an $L$-bad index gives $f(\ed(D))<f(\st(C))$.
Corollary~\ref{corollary:chain-length-d-minus-one} then implies
\begin{equation*}
    x+y\leq d+1+f(\ed(D))-f(\st(C))\leq d.
\end{equation*}
On the other hand, Equation~\eqref{equation:ed(D)-st(C)_leq_x+y-d-2} gives $x+y\geq d+1$, a contradiction.
Thus $\st(C)-\ed(D)\geq2$.

Again, since $P(C)$ intersects $Q(D)$, the left endpoint of $P(C)$ is at most the right endpoint of $Q(D)$. Thus
\begin{equation*}
    \st(C)-d\leq \ed(D)+y+1.
\end{equation*}
Consequently, we have
\begin{equation*}
    \st(C)-\ed(D)\leq d+y+1\leq 2d,
\end{equation*}
where the last inequality follows from $y\leq d-1$.
Lemma~\ref{lemma:chain-interaction} therefore gives $x+y\leq d-1$.
Finally, $Q(D)\ne\varnothing$ implies that $D$ is long and $2y\geq d$. Together with $x+y\leq d-1$, this gives $x\leq(d-2)/2$, so $C$ is short.
\end{proof}

We next show that no proxy can overlap two distinct proxies.

\begin{lemma} \label{lemma:at_most_one_overlap}
    Every proxy overlaps with at most one other proxy.
\end{lemma}
\begin{proof}
Suppose first that a proxy $P(C)$ intersects two proxies $Q(D_1)$ and $Q(D_2)$, where $\st(D_1)<\st(D_2)$. 
Let $x:=\len(C)$, $y_1:=\len(D_1)$ and $y_2:=\len(D_2)$. By
Lemma~\ref{lemma:collision_P(C)_Q(D)}, $C$ is short, so
$|P(C)|=3x+1$, and both $D_1$ and $D_2$ are long.
Recall that $Q(D_1)$ ends at $\ed(D_1) + y_1 + 1$ and $Q(D_2)$ starts at $\ed(D_2)+d-y_2+1$.
An interval meeting
both $Q(D_1)$ and $Q(D_2)$ must have size at least
\begin{equation*}
    1 + \left(\ed(D_2)+d-y_2+1\right) - \left( \ed(D_1) + y_1 + 1 \right) \geq 2d-y_1+1,
\end{equation*}
where the inequality holds since $\ed(D_2) - y_2 \geq \ed(D_1) + d$, by maximality.
Hence $3x\geq 2d-y_1$.

On the other hand, Lemma~\ref{lemma:collision_P(C)_Q(D)} gives
$x+y_1\leq d-1$, and therefore
\begin{equation*}
    y_1 \leq d-1 - \frac{2d - y_1}{3},
\end{equation*}
which implies $y_1\leq(d-3)/2$, contradicting that $D_1$ is long.

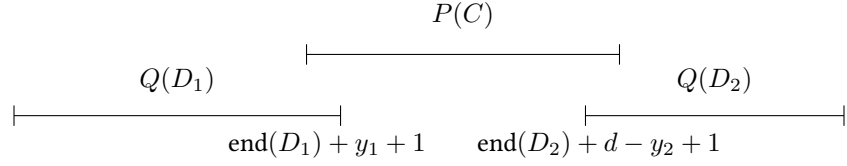
\begin{figure}[htbp]
\centering
\begin{tikzpicture}[x=0.9cm,y=0.9cm,every node/.style={font=\small}]

% top interval: P(C)
\draw (1.7,-0.5) -- (6.3,-0.5);
\draw (1.7,-0.35) -- (1.7,-0.65);
\draw (6.3,-0.35) -- (6.3,-0.65);

\node[above=6pt] at (4,-0.5) {$P(C)$};
% \node[above=2pt] at (1.5,0) {$\st(C)-d$};
% \node[above=2pt] at (8,0) {$\st(C)-d-1+\min\{3x+1,d+x\}$};

% bottom left interval: Q(D_1)
\draw (-2.6,-1.4) -- (2.2,-1.4);
\draw (-2.6,-1.25) -- (-2.6,-1.55);
\draw (2.2,-1.25) -- (2.2,-1.55);

\node[below=6pt] at (-0.2,-0.3) {$Q(D_1)$};
% \node[below=2pt] at (-0.8,-1.4) {$\ed(D_1)+d-y_1+1$};
\node[below=2pt] at (2,-1.4) {$\ed(D_1)+y_1+1$};

% bottom right interval: Q(D_2)
\draw (5.8,-1.4) -- (9.6,-1.4);
\draw (5.8,-1.25) -- (5.8,-1.55);
\draw (9.6,-1.25) -- (9.6,-1.55);

\node[below=6pt] at (7.7,-0.3) {$Q(D_2)$};
\node[below=2pt] at (6,-1.4) {$\ed(D_2)+d-y_2+1$};
% \node[below=2pt] at (9.6,-1.4) {$\ed(D_2)+y_2+1$};

\end{tikzpicture}
\caption{An illustration where $P(C)$ overlaps with both $Q(D_1)$ and $Q(D_2)$.}
\label{figure:PC-overlaps-QD1-QD2}
\end{figure}

Conversely, suppose that a proxy $Q(D)$ intersects two proxies
$P(C_1)$ and $P(C_2)$, where $\st(C_1)<\st(C_2)$. Let
$x_1:=\len(C_1)$ and $x_2:=\len(C_2)$. By
Lemma~\ref{lemma:collision_P(C)_Q(D)}, both $C_1$ and $C_2$ are
short, and $D$ is long. 
Similarly, an interval meeting both $P(C_1)$ and $P(C_2)$ must have size at least
\begin{equation}
    1 + \left( \st(C_2)-d \right) - \left( \st(C_1)-d+3x_1 \right)
    \geq d-2x_1+1,
    \label{equation:length_of_Q(D)_at_least_d-2x+1}
\end{equation}
where the inequality holds because
$\st(C_2)\geq\st(C_1)+x_1+d$.
However, since $D$ is long and $x_1 + \len(D) \leq d-1$, the proxy $Q(D)$ has size
\begin{equation*}
    2\cdot \len(D) - d + 1 \leq 2\cdot (d-1-x_1) - d + 1
    \leq d-2x_1-1,
\end{equation*}
contradicting Equation~\eqref{equation:length_of_Q(D)_at_least_d-2x+1}.
\end{proof}

We now combine these ingredients to prove Lemma~\ref{lemma:first-copy-packing}.

\begin{proofof}{Lemma~\ref{lemma:first-copy-packing}}
    Recall that our goal is to show that the total size of the first-track proxies is at most $n$.
    Lemma~\ref{lemma:same_type_proxy_disjoint} shows that only proxies of different types can overlap, and Lemma~\ref{lemma:at_most_one_overlap} implies that if two proxies $P(C)$ and $Q(D)$ overlap, they are disjoint from any other proxies.
    For each overlapping pair, we relocate $Q(D)$ so that all proxies become pairwise disjoint while preserving their total size.
    
    Consider an overlap between $P(C)$ and $Q(D)$, and let $x:=\len(C)$ and $y:=\len(D)$. By
    Lemma~\ref{lemma:collision_P(C)_Q(D)}, $C$ is short and $D$ is long. 
    When $C$ is short, its proxy $P(C)$ is a proper subinterval of the extended proxy $\tilde{P}(C)$.
    Specifically, $P(C)$ and $\tilde{P}(C)$ have the same starting position, while their ending positions differ by
    \begin{equation*}
        \left( \st(C)+x-1 \right) - \left( \st(C)-d+3x \right)
        = d - 2x - 1.
    \end{equation*}

    Since $D$ is long and $x+y \leq d-1$, we have
    \begin{equation*}
        |Q(D)| = 2y - d + 1 \leq 2\cdot (d-1-x) - d + 1 \leq d-2x-1.
    \end{equation*}

    Therefore, we can shift $Q(D)$ to immediately follow $P(C)$ inside $\tilde{P}(C)$.
    We do this for every overlapping pair of $P(C)$ and $Q(D)$.
    Since the extended proxies $\tilde{P}(C)$ are pairwise disjoint, it remains to show that $\tilde{P}(C)\setminus P(C)$ is disjoint from every $Q$-proxy other than $Q(D)$.

    Assume for contradiction that $Q(D')$, where $D' \neq D$, overlaps with $\tilde{P}(C) \setminus P(C)$.
    Then $Q(D')$ cannot have been relocated; otherwise, it would lie inside some $\tilde{P}(C')$, which is disjoint from $\tilde{P}(C)$.
    Since $Q(D)$ overlaps with $P(C)$ (before relocation), $Q(D')$ is disjoint from $P(C)$.
    Because the $Q$-proxies are pairwise disjoint and appear in the same order as their chains, $Q(D')$ can overlap $\tilde{P}(C)\setminus P(C)$ only if $D$ precedes $D'$. Its starting index must then lie in $\tilde{P}(C)\setminus P(C)$ (see Figure~\ref{figure:Q(D)_P(C)_Q(D')_tilde{P}(C)}).
    Let $y' := \len(D')$.

    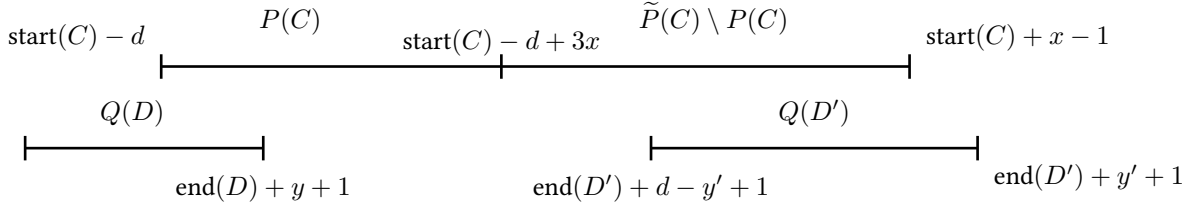
\begin{figure}[htbp]
    \centering
    \begin{tikzpicture}[
        x=0.9cm,
        y=0.9cm,
        line width=1pt,
        every node/.style={font=\small}
    ]
    %------------------------------------------------
    % Coordinates
    %------------------------------------------------
    \coordinate (A)  at (0,0);
    \coordinate (B)  at (5,0);
    \coordinate (C)  at (11,0);
    \coordinate (D1) at (-2,-1.2);
    \coordinate (D2) at (1.5,-1.2);
    \coordinate (E1) at (7.2,-1.2);
    \coordinate (E2) at (12,-1.2);
    %------------------------------------------------
    % Top interval: \widetilde{P}(C)
    %------------------------------------------------
    \draw (A) -- (C);
    \draw (A) ++(0,0.18) -- ++(0,-0.36);
    \draw (B) ++(0,0.18) -- ++(0,-0.36);
    \draw (C) ++(0,0.18) -- ++(0,-0.36);
    
    \node[above=8pt] at ($(A)!0.38!(B)$) {$P(C)$};
    \node[above=8pt] at ($(B)!0.52!(C)$) {$\widetilde{P}(C)\setminus P(C)$};
    
    \node[above left=2pt]  at (A) {$\st(C)-d$};
    \node[above]  at (B) {$\st(C)-d+3x$};
    \node[above right=2pt] at (C) {$\st(C)+x-1$};
    
    %------------------------------------------------
    % Bottom left interval: Q(D)
    %------------------------------------------------
    \draw (D1) -- (D2);
    \draw (D1) ++(0,0.18) -- ++(0,-0.36);
    \draw (D2) ++(0,0.18) -- ++(0,-0.36);
    \node[above=4pt] at ($(D1)!0.45!(D2)$) {$Q(D)$};
    \node[below=6pt] at (D2) {$\ed(D)+y+1$};
    %------------------------------------------------
    % Bottom right interval: Q(D')
    %------------------------------------------------
    \draw (E1) -- (E2);
    \draw (E1) ++(0,0.18) -- ++(0,-0.36);
    \draw (E2) ++(0,0.18) -- ++(0,-0.36);
    \node[above=4pt] at ($(E1)!0.5!(E2)$) {$Q(D')$};
    \node[below=6pt] at (E1) {$\ed(D')+d-y'+1$};
    \node[below right=2pt] at (E2) {$\ed(D')+y'+1$};
    \end{tikzpicture}
    \caption{An illustration of $Q(D)$, $P(C)$, $\tilde{P}(C)\setminus P(C)$, and $Q(D')$.}
    \label{figure:Q(D)_P(C)_Q(D')_tilde{P}(C)}
    \end{figure}
    
    By maximality, we have $\ed(D') - y' \geq \ed(D)+d$.
    Then the starting index of $Q(D')$ satisfies
    \begin{align*}
        \ed(D') + d - y' + 1 & \geq \ed(D) + 2d + 1 \\
        & \geq (\st(C) - d - y - 1) + 2d + 1 \\
        & = \st(C) + (d-y) \\
        & \geq \st(C) + x + 1,
    \end{align*}
    where the second inequality holds because the ending index of $Q(D)$ is at least the starting index of $P(C)$, and the last inequality follows from $x+y\leq d-1$.
    This is a contradiction, because the starting index of $Q(D')$ is strictly larger than the ending index of $\tilde{P}(C)$.
\end{proofof}

\section{Tightness and Impossibility Results}

In this section, we complement the competitive guarantee for \textsc{Flip} with two upper bounds.

\subsection{Tightness of \textsc{Flip}}

We first give an instance showing that the competitive analysis of \textsc{Flip} is tight for every $d\geq2$.

\begin{theorem}
    For every integer $d\geq 2$, there exists an instance on which
    \textsc{Flip} has competitive ratio $\frac{2d-1}{3d-2}$.
    In particular, this ratio converges to $2/3$ as $d\to\infty$.
\end{theorem}

\paragraph{Construction of the Hard Instance.}
Fix an integer $d\geq 2$, and consider an instance with
$3d-2$ offline vertices $U=\{u_1,\ldots,u_{3d-2}\}$
and equally many online vertices $V=\{v_1,\ldots,v_{3d-2}\}$.
The neighborhoods of the online vertices are defined by
\begin{equation*}
    N(v_i)=
    \begin{cases}
        \{u_i,u_{i+1},\ldots,u_{i+d-1}\}, & i\in\{1,\ldots,d-1\}, \\
        \{u_d,u_{d+1},\ldots,u_{2d-1}\}, & i\in\{d,\ldots,2d-1\}, \\
        \{u_{i-d+1},u_{i-d+2},\ldots,u_i\}, & i\in\{2d,\ldots,3d-2\}.
    \end{cases}
\end{equation*}
The left graph in Figure~\ref{figure:example-G-and-f} illustrates this construction for $d=4$.

Consider the arrival order
\begin{equation*}
    v_1,\ldots,v_{d-1}, \quad
    v_{2d},\ldots,v_{3d-2}, \quad
    v_d,\ldots,v_{2d-1}.
\end{equation*}
Under either matching rule used by \textsc{Flip}, all $2d-2$ vertices in the first two blocks are matched. At that point, exactly one vertex in $\{u_d,u_{d+1},\ldots,u_{2d-1}\}$ remains unmatched: it is $u_d$ under one rule and $u_{2d-1}$ under the other.
Since every vertex in the final block has this set as its neighborhood, the first such vertex is matched to the unique remaining offline vertex, while the other $d-1$ vertices remain unmatched.
Consequently,
\begin{equation*}
    \alg_{\mathrm{up}}
    =
    \alg_{\mathrm{down}}
    =
    2d-1.
\end{equation*}

On the other hand, the instance admits the perfect matching
$\{(u_i,v_i):i\in[3d-2]\}$, and hence
$\opt=3d-2$.
This gives the claimed upper bound.

\subsection{A \texorpdfstring{$3/4$}{} Upper Bound for Any Algorithm}

In this subsection, we show that no randomized online algorithm can
achieve a competitive ratio strictly larger than $3/4$.

\paragraph{Construction of the Hard Instance.}
Assume that $d$ is even. Let $k$ be an even integer and let $n=kd$.
Partition the offline set $U=\{u_1,\ldots,u_n\}$ into $k$ consecutive blocks $U_1,\ldots,U_k$, each of size $d$.
We further divide each block $U_i$ into two half-blocks:
\begin{align*}
    U_{i,1} &:=\{u_{(i-1)d+1},\ldots,u_{(i-1)d+d/2}\}, \\
    U_{i,2} &:=\{u_{(i-1)d+d/2+1},\ldots,u_{id}\},
\end{align*}
and write $U_i=U_{i,1}\cup U_{i,2}$.

The online vertices arrive in two phases.
The first phase is common to both realizations of the instance.
For every $i\in[k]$, let there be $d/2$ online vertices whose neighborhood is $U_i$.

After the first phase, one of the following two continuations is selected uniformly at random.

\begin{itemize}
    \item In the \emph{odd continuation}, for every
    $i\in\{1,3,\ldots,k-1\}$, there are $d$ online vertices whose neighborhood is $U_{i,2}\cup U_{i+1,1}$.

    \item In the \emph{even continuation}, there are $d/2$ online vertices whose neighborhood is $U_1$, $d/2$ online vertices
    whose neighborhood is $U_k$, and, for every $i\in\{2,4,\ldots,k-2\}$, there are $d$ online vertices whose neighborhood is $ U_{i,2}\cup U_{i+1,1}$.
\end{itemize}

Thus, each realization contains $n/2$ online vertices in each phase, and every online neighborhood is an interval of exactly $d$ consecutive offline vertices; see Figure~\ref{figure:three-quarter-hard-instance} for an illustration with $k=4$.

\begin{figure}[htbp]
\centering
\resizebox{0.96\textwidth}{!}{%
\begin{tikzpicture}[
    x=0.72cm,
    y=0.72cm,
    every node/.style={font=\scriptsize},
    firstphase/.style={draw=blue!70!black, line width=0.9pt},
    secondphase/.style={draw=red!70!black, line width=0.9pt}
]

% Odd continuation.
\begin{scope}
\node[font=\small] at (4,2.5) {odd continuation};
\node[left] at (-0.25,1.25) {Phase 1};
\node[left] at (-0.25,-0.95) {Phase 2};

\foreach \i in {1,...,4} {
    \pgfmathsetmacro{\xl}{2*(\i-1)}
    \draw[fill=blue!10] (\xl,0) rectangle ++(1,0.62);
    \draw[fill=blue!10] ({\xl+1},0) rectangle ++(1,0.62);
    \node at ({\xl+0.5},0.31) {$U_{\i,1}$};
    \node at ({\xl+1.5},0.31) {$U_{\i,2}$};
    \node[below=2pt] at ({\xl+1},0) {$U_\i$};

    \draw[firstphase] (\xl,1.25) -- ({\xl+2},1.25);
    \draw[firstphase] (\xl,1.39) -- (\xl,1.11);
    \draw[firstphase] ({\xl+2},1.39) -- ({\xl+2},1.11);
    \node[above=2pt] at ({\xl+1},1.25) {$d/2$};
}

\draw[secondphase] (1,-0.95) -- (3,-0.95);
\draw[secondphase] (1,-0.81) -- (1,-1.09);
\draw[secondphase] (3,-0.81) -- (3,-1.09);
\node[below=2pt] at (2,-0.95) {$d$};

\draw[secondphase] (5,-0.95) -- (7,-0.95);
\draw[secondphase] (5,-0.81) -- (5,-1.09);
\draw[secondphase] (7,-0.81) -- (7,-1.09);
\node[below=2pt] at (6,-0.95) {$d$};
\end{scope}

\draw[densely dashed,gray] (9,-1.45)--(9,2.35);

% Even continuation.
\begin{scope}[shift={(10,0)}]
\node[font=\small] at (4,2.5) {even continuation};
\node[left] at (-0.25,1.25) {Phase 1};
\node[left] at (-0.25,-0.95) {Phase 2};

\foreach \i in {1,...,4} {
    \pgfmathsetmacro{\xl}{2*(\i-1)}
    \draw[fill=blue!10] (\xl,0) rectangle ++(1,0.62);
    \draw[fill=blue!10] ({\xl+1},0) rectangle ++(1,0.62);
    \node at ({\xl+0.5},0.31) {$U_{\i,1}$};
    \node at ({\xl+1.5},0.31) {$U_{\i,2}$};
    \node[below=2pt] at ({\xl+1},0) {$U_\i$};

    \draw[firstphase] (\xl,1.25) -- ({\xl+2},1.25);
    \draw[firstphase] (\xl,1.39) -- (\xl,1.11);
    \draw[firstphase] ({\xl+2},1.39) -- ({\xl+2},1.11);
    \node[above=2pt] at ({\xl+1},1.25) {$d/2$};
}

\draw[secondphase] (0,-0.95) -- (2,-0.95);
\draw[secondphase] (0,-0.81) -- (0,-1.09);
\draw[secondphase] (2,-0.81) -- (2,-1.09);
\node[below=2pt] at (1,-0.95) {$d/2$};

\draw[secondphase] (3,-0.95) -- (5,-0.95);
\draw[secondphase] (3,-0.81) -- (3,-1.09);
\draw[secondphase] (5,-0.81) -- (5,-1.09);
\node[below=2pt] at (4,-0.95) {$d$};

\draw[secondphase] (6,-0.95) -- (8,-0.95);
\draw[secondphase] (6,-0.81) -- (6,-1.09);
\draw[secondphase] (8,-0.81) -- (8,-1.09);
\node[below=2pt] at (7,-0.95) {$d/2$};
\end{scope}

\end{tikzpicture}%
}
\caption{The odd and even continuations for $k=4$. Each box represents a half-block of $d/2$ consecutive offline vertices. A horizontal segment represents the common neighborhood of a group of online vertices, and its label gives the number of vertices in that group. The blue first-phase groups are common to both realizations; the red second-phase groups distinguish the two continuations.}
\label{figure:three-quarter-hard-instance}
\end{figure}
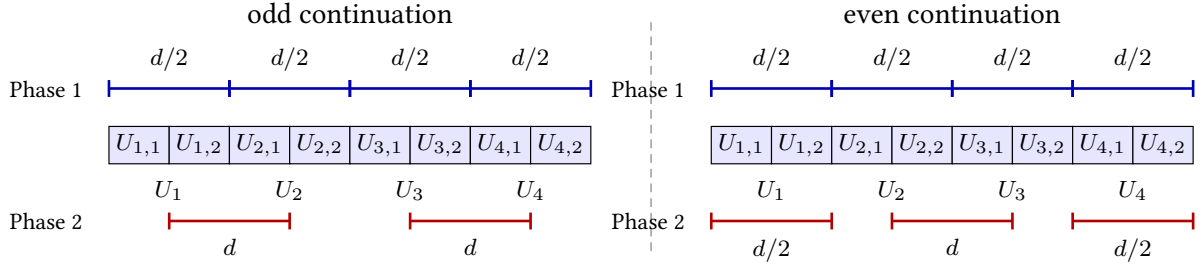

\begin{theorem}
\label{theorem:3/4-problem-hardness}
    No randomized online algorithm has a competitive ratio strictly larger than $3/4$.
\end{theorem}
\begin{proof}
    We first analyze an arbitrary deterministic online algorithm $\mathcal{A}$.
    For each $i\in [k]$, let $a_i$ and $b_i$ denote the numbers of vertices in $U_{i,1}$ and $U_{i,2}$ that are matched by $\mathcal A$ during the first phase.
    Let
    \begin{equation*}
        m := \sum_{i=1}^k (a_i+b_i).
    \end{equation*}
    Since the first phase contains $n/2$ online vertices, $m\leq n/2$.

    Let $\mathcal{A}_{\mathrm{odd}}$ and $\mathcal{A}_{\mathrm{even}}$ be the matching sizes obtained under the two continuations, respectively.
    In the odd continuation, the number of offline vertices still available in $U_{i,2}\cup U_{i+1,1}$ is $d-b_i-a_{i+1}$.
    Consequently,
    \begin{equation*}
        \mathcal{A}_{\mathrm{odd}} \leq m+\sum_{\substack{1\leq i\leq k-1 \\ i\text{ odd}}}
        ( d-b_i-a_{i+1} ).
    \end{equation*}
    Similarly, the two endpoint groups in the even continuation contain $d$ online vertices in total, while each internal group can match at most the number of available vertices in its neighborhood.
    Hence,
    \begin{equation*}
        \mathcal{A}_{\mathrm{even}} \leq m+d+\sum_{\substack{2\leq i\leq k-2 \\ i\text{ even}}} ( d-b_i-a_{i+1} ).
    \end{equation*}
    Averaging these two inequalities gives
    \begin{align*}
        \frac{\mathcal{A}_{\mathrm{odd}} +\mathcal{A}_{\mathrm{even}}}{2} &\leq m+\frac{1}{2} \left( n-\sum_{i=1}^{k-1}(b_i+a_{i+1}) \right) = m+\frac{1}{2} (n-m+a_1+b_k) \\
        &= \frac{n}{2}+\frac{m}{2}+\frac{a_1+b_k}{2} \leq \frac{3n}{4}+\frac{d}{2} = \left(\frac{3}{4}+\frac{1}{2k}\right) \cdot n,
    \end{align*}
    where the last inequality follows from $m\leq n/2$ and $a_1,b_k\leq d/2$.

    On the other hand, both realizations admit a perfect matching:
    \begin{itemize}
        \item For the odd continuation, match the second-phase vertices to $U_{i,2}\cup U_{i+1,1}$ for every odd $i$, and match the first-phase vertices associated with $U_i$ to $U_{i,1}$ when $i$ is odd and to $U_{i,2}$ when $i$ is even.

        \item For the even continuation, match the two endpoint groups to $U_{1,1}$ and $U_{k,2}$, respectively.
        Match every internal second-phase group to its entire neighborhood, and match each first-phase group to the remaining half of its block. 
    \end{itemize}
    Therefore, we have $\opt=n$ for both realizations.

    The preceding bound holds for every deterministic algorithm under the uniform distribution over the two continuations. By Yao's minimax principle, for every randomized online algorithm, at least one realization in the support satisfies
    \begin{equation*}
        \frac{\E[\alg]}{\opt} \leq \frac{3}{4}+\frac{1}{2k}.
    \end{equation*}
    Since $k$ can be arbitrarily large, the desired upper bound of $3/4$ follows.
\end{proof}

\section{Without the Uniform-Length Assumption}
\label{section:1-1/e}

In this section, we show that convexity alone is insufficient for beating the competitive ratio of $1-1/e$.
More precisely, we prove the following theorem.

\begin{theorem}
\label{theorem:1-1/e}
    No algorithm has a competitive ratio strictly larger than $1-1/e$ for online matching in convex bipartite graphs.
\end{theorem}

It suffices to prove the upper bound for deterministic fractional algorithms: the edge-matching probabilities of any randomized algorithm induce a deterministic fractional algorithm with the same ratio.

Recall the classic upper-triangular instance of Karp et al.~\cite{conf/stoc/KarpVV90} for general online bipartite matching.
Starting with all offline vertices, the adversary repeatedly introduces one online vertex adjacent to the current set and then removes the least-matched offline vertex from all future neighborhoods.
Unfortunately, the resulting graph is not convex: the removal of the least-matched offline vertex splits the offline vertices into two disconnected segments.
Our main idea is to replace each vertex with sufficiently many copies, so that the neighborhoods are still intervals after some offline vertices are removed.
% Our construction applies the same idea recursively at the interval level: an interval of size $k!$ is partitioned into $k$ consecutive subintervals of size $(k-1)!$, a group of $(k-1)!$ common online neighbors is introduced, the least-matched subinterval is dropped, and the construction continues within every remaining subinterval.
% As in the instance of Karp et al.~\cite{conf/stoc/KarpVV90}, uniform fractional allocation is optimal, yielding the $1-1/e$ bound.

\subsection{The Hard Instance}

Let $\mathcal{A}$ be an arbitrary deterministic fractional algorithm.
Fix an integer $m\geq1$, and let $U$ be an ordered set of $m!$ offline vertices.
For every $k\in [m]$ and every interval $I\subseteq U$ of size $k!$, we define the following recursive procedure $\mathsf{Build}(I,k)$.
The hard instance is generated by $\mathsf{Build}(U,m)$.

\begin{itemize}
    \item If $k=1$, the procedure introduces one online vertex whose neighborhood is the unique vertex in $I$.

    \item For $k\geq 2$, partition $I$ into $k$ consecutive subintervals $I_1,I_2,\ldots,I_k$, each containing $(k-1)!$ offline vertices.
    Introduce $(k-1)!$ online vertices, each with neighborhood $I$.
    We refer to them as the level-$k$ online vertices.
    After $\mathcal{A}$ processes these vertices, let $\lambda_u$ denote the current matched fraction of each $u\in I$.
    Choose an index $r\in[k]$ such that
    \begin{equation*}
        \sum_{u\in I_r}\lambda_u = \min_{j\in[k]} \sum_{u\in I_j}\lambda_u,
    \end{equation*}
    breaking ties arbitrarily.
    No future online vertex is adjacent to any vertex in $I_r$; we call $I_r$ the \emph{dropped interval}.
    For every $j\neq r$, recursively apply $\mathsf{Build}(I_j,k-1)$.
\end{itemize}

Every online neighborhood produced by the construction is an interval in the given order of $U$.
Moreover, the instance admits a perfect matching of size $m!$.
Indeed, at every recursive call, the level-$k$ online vertices can be matched to the dropped interval $I_r$, whose vertices have no future online neighbors.
Therefore, we have $\opt = m!$.

\subsection{Bounding the Fractional Matching}

We upper-bound the fractional matching size of $\mathcal{A}$ by tracking the (average) matched fraction within each recursive interval.

First, suppose that $\mathcal{A}$ distributes every online level uniformly over its neighboring interval.
Suppose an invocation $\mathsf{Build}(I,k)$ begins with matched fraction $x$ on every vertex of $I$, and let $w_k(x)$ denote the final average matched fraction after matching the online vertices at level $k$.
Clearly, we have $w_1(x)=1$.
For $k\geq2$, the level-$k$ online vertices raise the average matched fraction in $I$ to $\hat x=\min\{1,x+1/k\}$.
One of the $k$ subintervals is then dropped with average matched fraction $\hat x$, while the other $k-1$ subintervals continue with the same average matched fraction.
Therefore,
\begin{equation*}
    w_k(x) = \frac{\hat x}{k} + \frac{k-1}{k}\cdot w_{k-1}(\hat x).
\end{equation*}
Unrolling this recurrence gives
\begin{equation*}
    w_k(x) = \frac{1}{k} \cdot \sum_{l=1}^k \min \left\{ 1, x+\sum_{s=l}^k \frac{1}{s} \right\}.
\end{equation*}
Notice that $w_k$ is concave and non-decreasing.
Moreover, $w_k(x) - x$ is non-increasing.

We next show that $w_k(x)$ upper-bounds the final average matched fraction in $\mathsf{Build}(I,k)$ even when the allocation across subintervals is nonuniform.

\begin{lemma}
\label{lemma:uniform-allocation-is-optimal}
    Consider an invocation $\mathsf{Build}(I,k)$ that begins with average matched fraction $x$ over the vertices in $I$.
    The final average matched fraction in $I$ is at most $w_k(x)$.
\end{lemma}
\begin{proof}
    We prove by induction on $k\geq 1$.
    The claim is immediate for $k=1$, since the final matched fraction of the unique offline vertex is at most $1=w_1(x)$.

    Now suppose that $k\geq2$.
    After the level-$k$ online vertices have been processed, let $\bar y$ and $y_j$ be the average matched fraction in $I$ and $I_j$ (where $j\in [k]$), respectively.
    Since level $k$ contains $(k-1)!$ online vertices whereas $|I|=k!$, it can increase the average load in $I$ by at most $1/k$.
    Hence, $\bar y \leq \min\left\{1, x+1/k\right\} = \hat x$.
    
    Let $I_r$ be the dropped subinterval.
    By definition, we have $y_r = \min_{j\in [k]} y_j \leq \bar y$.
    
    The dropped subinterval $I_r$ contributes $(k-1)!\cdot y_r$ to the final matching.
    By the induction hypothesis, every remaining $k-1$ subinterval $I_j$ contributes at most $(k-1)!\cdot w_{k-1}(y_j)$.
    Thus, the final average matched fraction in $I$ is at most
    \begin{equation}
    \label{equation:ub-average-matched-fraction-in-I}
        \frac{1}{k} \left( y_r+\sum_{j\neq r} w_{k-1}(y_j) \right).
    \end{equation}
    
    By the concavity of $w_{k-1}$ and $y_r\leq \bar y$,
    \begin{align*}
        y_r + \sum_{j\neq r} w_{k-1}(y_j) &= y_r + \sum_{j=1}^k w_{k-1}(y_j) - w_{k-1}(y_r) \\
        &\leq \bar y + \sum_{j=1}^k w_{k-1}(y_j) - w_{k-1}(\bar y)
        \leq \bar y + (k-1)\cdot w_{k-1}(\bar y),
    \end{align*}
    where the first inequality holds because $w_{k-1}(x)-x$ is non-increasing. 
    The last expression is non-decreasing in $\bar y$. Therefore, the quantity in Equation~\eqref{equation:ub-average-matched-fraction-in-I} is at most
    \begin{equation*}
        \frac{\hat x}{k} + \frac{k-1}{k}\cdot w_{k-1}(\hat x) = w_k(x),
    \end{equation*}
    which completes the induction.
\end{proof}

Now we can prove Theorem~\ref{theorem:1-1/e}.

\begin{proofof}{Theorem~\ref{theorem:1-1/e}}
    Consider the instance generated by $\mathsf{Build}(U,m)$.
    Initially, every offline vertex has zero matched fraction.
    Recall that $\opt=m!$.
    By Lemma~\ref{lemma:uniform-allocation-is-optimal}, the competitive ratio of any deterministic fractional algorithm in this instance is upper-bounded by
    \begin{equation*}
        w_m(0) = \frac{1}{m} \sum_{\ell=1}^m \min\left\{ 1, \sum_{s=\ell}^m\frac{1}{s} \right\}.
    \end{equation*}

    As $m\to\infty$, view $\ell/m$ as $x\in(0,1]$. The harmonic sum $\sum_{s=\ell}^m 1/s$ converges to $-\ln x$, and the preceding expression converges to
    \begin{align*}
        \lim_{m\to\infty} w_m(0) = \int_0^1 \min\{1,-\ln x\} \mathrm{d}x 
        = \int_0^{1/e}1 \mathrm{d} x + \int_{1/e}^1 (-\ln x) \mathrm{d}x = 1-\frac{1}{e},
    \end{align*}
    and the theorem follows.
\end{proofof}

\section*{Declaration for the Use of AI} 
The main algorithmic framework, including the design of the algorithm and the reduction to bad indices in the offset function, was discovered fully by humans.
The authors acknowledge the use of GPT-5.6 Sol in the exploration for part of the proofs. All arguments were independently verified by the authors. The manuscript was written entirely by the authors, who take full responsibility for its content.

\section*{Acknowledgments} 
We thank Siddhartha Banerjee and Billy Jin for posing the question of determining the optimal competitive ratio in the setting of Theorem~\ref{theorem:arbitrary-length} to us, and for subsequent helpful discussions. Part of this work was done while Yilong Feng and Kangning Wang were visiting Shanghai University of Finance and Economics.

\cleardoublepage
\phantomsection
\addcontentsline{toc}{section}{Bibliography}

\bibliography{ref}
\bibliographystyle{alpha}

\end{document}